\pdfoutput=1
\documentclass[acmsmall,screen]{acmart}

\setcopyright{rightsretained}
\acmJournal{PACMPL}
\acmYear{2019} \acmVolume{3} \acmNumber{POPL} \acmArticle{44} \acmMonth{1} \acmPrice{}\acmDOI{10.1145/3290357}
\copyrightyear{2019}

\usepackage{booktabs}   
\usepackage{subcaption} 

\usepackage{mathtools}
\usepackage{xparse}
\usepackage{mathpartir}
\usepackage{etoolbox}
\usepackage{mathpartir}
\usepackage{marvosym}

\newcommand\ifempty[3]{\expandafter\ifstrempty{#1}{#2}{#3}}
\newcommand\ifnempty[2]{\ifempty{#1}{}{#2}}

\newcommand\ifthenabove[1]{%
    \ifnempty{#1}{^{#1}}%
}

\makeatletter
\newcommand\mpar@math[3]{\ensuremath{\left(#1\right)\ifnempty{#2}{_{#2}}\ifnempty{#3}{^{#3}}}}
\newcommand\mpar@text[3]{\ensuremath{(#1)\ifnempty{#2}{_{#2}}\ifnempty{#3}{^{#3}}}}
\newcommand\mpar@belowabove[3]{\mathchoice{\mpar@math{#1}{#2}{#3}}{\mpar@text{#1}{#2}{#3}}{\mpar@text{#1}{#2}{#3}}{\mpar@text{#1}{#2}{#3}}}

\def\mpar@grabbelow#1#2#3_#4{%
  \mpar@grab{\ifnempty{#1}{#1, }#4}{#2}{#3}%
}
\def\mpar@grabbove#1#2#3^#4{%
  \mpar@grab{#1}{\ifnempty{#2}{#2, }#4}{#3}%
}
\newcommand\mpar@grab[3]{%
  \@ifnextchar_{\mpar@grabbelow{#1}{#2}{#3}}%
  {\@ifnextchar^{\mpar@grababove{#1}{#2}{#3}}%
  {\mpar@belowabove{#3}{#1}{#2}}}%
}

\newcommand\mpar{\mpar@grab{}{}}
\makeatother

\newcommand\funp{\mpar}
\newcommand\funu[1]{\ifnempty{#1}{\funp{#1}}}
\newcommand\fund[2]{\ifnempty{#1#2}{\funp{\sep{#1}{#2}}}}
\newcommand\funt[3]{\ifnempty{#1#2#3}{\funp{\sept{#1}{#2}{#3}}}}

\newcommand\sep[2]{#1\ifnempty{#1}{\ifnempty{#2}{,}}#2}
\newcommand\sept[3]{\sep{#1}{#2}\ifnempty{#1#2}{\ifnempty{#3}{,}}#3}

\newcommand{\eqdef}{\triangleq}

\usepackage{todonotes}

\newcommand\cons[2]{#1; #2}
\newcommand\emptylist{[]}

\newcommand\mset[2][]{%
    \ifempty{#2#1}%
	    {\emptyset}%
	    {\left\{%
            #1%
            \ifnempty{#1}{\middle|}%
		    #2%
	    \right\}}%
}

\newcommand\st{\,\middle|\,}

\newcommand\myvecbase[3][]{\ifnempty{#1}{#1\funu}{#2{1}}..\ifnempty{#1}{#1\funu}{#2{#3}}}

\newcommand\mandatorytofacultative[2]{#1[#2]{}}
\newcommand\myvec[2][]{\myvecbase[#1]{\mandatorytofacultative{#2}}}
\newcommand\myvecp[3][]{\mpar{\myvec[#1]{#2}{#3}}}
\newcommand\myvecs[3][]{\mset{\myvec[#1]{#2}{#3}}}

\newcommand\mandatorytoindex[2]{#1\ifnempty{#2}{_{#2}}}
\newcommand\myveci[2][]{\myvecbase[#1]{\mandatorytoindex{#2}}}
\newcommand\myvecip[3][]{\mpar{\myveci[#1]{#2}{#3}}}

\makeatletter
\newcommand\myvecmap@fun[6]{\ifnempty{#1}{#1\funu}{#2[#6]{}} #3 \ifnempty{#4}{#4\funu}{#5[#6]{}}}
\newcommand\myvecmap[6]{%
    \myvecbase{\myvecmap@fun{#1}{#2}{#3}{#4}{#5}}{#6}}
\makeatother

\newcommand\dom{\mathit{dom}\funu}

\newcommand{\tvar}[1][]{x_{t\ifnempty{#1}{_{#1}}}}
\newcommand{\fvar}[1][]{x_{f\ifnempty{#1}{_{#1}}}}
\DeclareDocumentCommand{\svar}{ m O{} }{#1{\ifnempty{#2}{_{#2}}}}
\newcommand{\ivar}{x_\sigma}
\newcommand{\ovar}{x_o}

\newcommand\Sort[1][\typenv]{\ensuremath{\mathit{Sort}_{#1}\funu}}
\newcommand\sort[1][]{\ensuremath{s\ifnempty{#1}{_{#1}}}}
\newcommand\sortp[1][]{\ensuremath{s'\ifnempty{#1}{_{#1}}}}
\newcommand\Tvar{\ensuremath{\mathit{Tvar}\funu}}
\newcommand\sig{\ensuremath{\mathit{sig}\funu}}
\newcommand{\termenv}{E}
\newcommand{\term}[1][]{t\ifnempty{#1}{_{#1}}}

\newcommand{\progpoint}{\texttt{pp}}
\newcommand\ppe{\epsilon}
\newcommand\ppc[3][]{\ifnempty{#2}{#2}{\cdot}\ifnempty{#3}{\ifempty{#1}{#3}{#3[#1]}}}
\newcommand\cppv[3][]{\ifnempty{#2}{#2}{\rule[0.5ex]{2pt}{0.5pt}}\ifnempty{#3}{\ifempty{#1}{#3}{#3[#1]}}}
\newcommand\syntaxcontraint[1]{\left[#1\right]}
\newcommand\subterm[2]{#1@#2}
\newcommand{\ppenv}{\mathit{HPP}}

\newcommand{\bvar}{V}

\newcommand\rulen[1]{\textsc{#1}}
\newcommand\Rule[3]{\ensuremath{\rulen{#1}{\mpar{#2}} \coloneq #3}}
\newcommand\Rulea[3]{\ensuremath{\rulen{#1}{\mpar{#2}} & \coloneq #3}} 
\newcommand{\Rules}{\ensuremath{\mathit{Rules}}}

\newcommand\derivH{H}
\newcommand\deriv[3]{\derivH\funu{#1\ifnempty{#1}{,}#2,#3}}
\newcommand\branches[2][\bvar]{\mpar{#2}_{#1}}
\newcommand\branchesV{\branches}
\newcommand{\eqp}{\mathrel{?\!\vartriangleright}}
\newcommand\filter[3]{#1\funu{#2}\ifnempty{#3}{\eqp #3}}

\newcommand{\skel}[1][]{S{\ifnempty{#1}{_{#1}}}}

\newcommand\litS{\ensuremath{\mathit{lit}}}
\newcommand\identS{\ensuremath{\mathit{ident}}}
\newcommand\exprS{\ensuremath{\mathit{expr}}}
\newcommand\statS{\ensuremath{\mathit{stat}}}

\newcommand\jsid[2][]{{\textup{\texttt{\small #2}}\ifthenabove{#1}}}

\newcommand\constn{\mathit{const}} 
\newcommand\const[1]{\constn\funu{#1}}
\newcommand\varn{\mathit{var}} 
\newcommand\var[1]{\varn\funu{#1}}
\newcommand\access[1]{!#1}
\newcommand{\erefn}{\ensuremath{\mathit{ref}}}
\newcommand\eref[1]{\erefn\funu{#1}}
\newcommand{\ein}{\ensuremath{\mathit{in}}}
\newcommand{\eoutn}{\ensuremath{\mathit{out}}}
\newcommand{\eout}[1]{\eoutn\mpar{#1}}
\newcommand\sskip{\mathit{skip}}
\newcommand\sthrow{\mathit{throw}}
\newcommand\iname{\mathit{if}}
\newcommand\wname{\mathit{while}}
\newcommand\asnescape[2]{#1\ {:=}\ #2}
\newcommand{\heapwrite}[2]{#1 \leftarrow #2}
\newcommand\seq[2]{{#1};\ #2}
\newcommand\stry[2]{\mathit{try} \, #1 \, \mathit{catch} \, #2}
\newcommand\sif[3]{\iname\ifnempty{#1#2#3}{\,#1\,#2\ifnempty{#3}{\,#3}}}
\newcommand\while[2]{\wname\ifnempty{#1#2}{\,#1\,#2}}

\newcommand{\interpIn}{\Sigma}
\newcommand{\interpOut}{O}
\newcommand\vdefined[1]{{#1\downarrow}}
\DeclareDocumentCommand{\interpDef}{ O{\interpIn} O{I} m }{\ensuremath{\vdefined{\left\llbracket #1, #3 \right\rrbracket^{#2}}}}
\DeclareDocumentCommand{\interp}{ O{\interpIn} O{I} m O{\interpOut}}{\ensuremath{%
    \left\llbracket #3 \right\rrbracket\ifnempty{#2}{^{#2}}\ifnempty{#1}{\mpar{#1}}%
    \ifnempty{#4}{ \Downarrow #4}}}

\DeclareDocumentCommand{\interpbDef}{ O{\interpIn} O{I} m O{\bvar}
}{\ensuremath{\vdefined{\left\llbracket #3
\right\rrbracket^{#2}_{#4}\mpar{#1}}}}
\DeclareDocumentCommand{\interpb}{ O{\interpIn} O{I} m O{\interpOut}
  O{\bvar}}{\ensuremath{\left\llbracket #3 \right\rrbracket\ifnempty{#2}{^{#2}}\ifnempty{#5}{_{#5}}\mpar{#1}
    \Downarrow #4}}

\DeclareDocumentCommand{\merge}{ O{I} m m O{\interpIn'} O{\bvar}
}{\ensuremath{\interpb[#2][#1]{\bigoplus_{#3}}[#4][#5]}}
\newcommand{\outfun}{\mathcal{O}}

\newcommand{\csem}{\ensuremath{\Downarrow}}
\newcommand{\asem}{\ensuremath{\abs{\csem}}}

\DeclareDocumentCommand{\typvar}{ m O{} }{\ensuremath{\tau\ifnempty{#1}{#1}\ifnempty{#2}{_{#2}}}}
\newcommand\typenv{\Gamma}
\newcommand\extsymbol{ + }
\newcommand\extenv[3]{#1\extsymbol#2\mapsto#3}

\newcommand\envrestr[2][\bvar]{\left.#2\right|_{#1}} 
\newcommand\typin{\mathit{in}\funu}
\newcommand\typout{\mathit{out}\funu}
\newcommand{\typsig}{\mathit{fsort}\funu}

\newcommand{\dfvarset}{\ensuremath{\mathcal{D}}}
\DeclareDocumentCommand{\wfinterp}{ O{\typenv,\dfvarset} m
  O{\mpar{\typenv',\dfvarset'}} O{} }{\interpb[#1][\texttt{wf}]{#2}[#3][#4]}

\newcommand\valT{\mathit{val}}
\newcommand\intT{\mathit{int}}
\newcommand\boolT{\mathit{bool}}
\newcommand\locT{\mathit{loc}}
\newcommand\stateT{\mathit{store}}
\newcommand\heapT{\mathit{heap}}
\newcommand\IOstateT{\mathit{state}}
\newcommand\valIOstateT{\mathit{valState}}
\newcommand\excIOstateT{\mathit{excState}}
\newcommand\inT{\mathit{in}}
\newcommand\outT{\mathit{out}}
\newcommand\unit{()}

\newcommand{\concrenv}{\Sigma}
\newcommand{\tripleset}{T}

\DeclareDocumentCommand{\concrinterp}{ O{\concrenv,\tripleset} m
  O{\mpar{\concrenv',\tripleset}} O{} }{ \interpb[#1][]{#2}[#3][#4] }

\DeclareDocumentCommand{\filterinterp}{ O{I} m m O{\Downarrow} m }{\ensuremath{\left\llbracket #2
\right\rrbracket^{#1}\ifnempty{#3}{\funu{#3}\ifnempty{#5}{#4}}#5}}

\newcommand{\concrfilter}[3]{\filterinterp[]{#1}{#2}{#3}}

\newcommand{\cinterpf}[1][]{\immediateconsequence\ifnempty{#1}{^{#1}}\funu}

\newcommand{\cval}[1][]{v\ifnempty{#1}{_{#1}}}
\newcommand{\cvalp}[1][]{v'\ifnempty{#1}{_{#1}}}

\newcommand{\aterm}[1][]{\term^{\#}\ifnempty{#1}{_{#1}}}
\newcommand{\atermp}[1][]{\term^{\#\prime}\ifnempty{#1}{_{#1}}}

\newcommand{\aval}[1][]{v^{\#}\ifnempty{#1}{_{#1}}}
\newcommand{\avalp}[1][]{v\ifnempty{#1}{_{#1}}^{\#\prime}}

\newcommand{\cstate}[1][]{\sigma\ifnempty{#1}{_{#1}}}
\newcommand{\astate}[1][]{\sigma^{\#}\ifnempty{#1}{_{#1}}}

\DeclareDocumentCommand{\abstrinterp}{ O{\absflag} O{\abstrenv,\abstripleset}  m
  O{\absflag} O{\abstrenvp,\abstripleset} O{} }{ \interpb[#1,#2][\#]{#3}[\mpar{#4,#5}][#6] }

\newcommand{\concr}{\gamma}

\newcommand{\abstrenv}[1][]{\Sigma^\#\ifnempty{#1}{_{#1}}}
\newcommand{\abstrenvp}[1][]{\Sigma^{\#\prime}\ifnempty{#1}{_{#1}}}

\newcommand\abstripleset[1][]{{T^\#\ifnempty{#1}{_{#1}}}} 
\newcommand{\abstrfilter}[3]{\filterinterp[\#]{#1}{#2}[=]{#3}}

\newcommand\ainterpf[1][]{{\immediateconsequence^\#}\ifnempty{#1}{^{#1}}\funu}

\newcommand{\stSplit}{\mathit{Sp}\funu}
\newcommand\sinterpf[1]{\stSplit{\ainterpf{\stSplit{#1}}}}

\makeatletter
\def\abot@complete#1[#2]{\bot{\ifempty{#2}{{\ifnempty{#1}{_{#1}}}}{_{#1[#2]}}}}
\newcommand\abot[1]{\@ifnextchar[{\abot@complete{#1}}{\abot@complete{#1}[]}}
\makeatother

\newcommand\OKst[1][]{\mathit{OKst}\ifnempty{#1}{_{#1}}\fund}
\newcommand\OKout[1][]{\mathit{OKout}\ifnempty{#1}{_{#1}}\fund}

\newcommand{\ppgen}[1][t_0]{\ppenv_{#1}\funt}
\newcommand{\eqset}{\mathcal{C}}

\newcommand{\eqdfvar}[1][N]{\mathcal{D}_{\rulen{#1}}\funu}
\newcommand{\eqgen}[3][\rulen{N},\progpoint,\eqset]{\interp[#1][c]{#2}[#3]}
\newcommand{\eqgenb}[3][\rulen{N},\progpoint,\eqset]{\interpb[#1][c]{#2}[#3]}
\newcommand{\eqfilter}[3]{\filterinterp[c]{#1}{#2}[=]{#3}}
\newcommand\Gen{\mathit{Gen}\funu}
\newcommand\Gent{\mathit{PP}\funu}
\newcommand{\solution}{\mathcal{S}}

\makeatletter 
\newcommand\@absstarnothing[1]{\ensuremath{\mathrel{#1^{\#}}}}
\newcommand\@absnostarnothing[1]{\ensuremath{{#1^{\#}}}}
\def\@absstarbrack#1[#2]{\ensuremath{\mathrel{#1^{\#}\left[#2\right]}}}
\def\@absnostarbrack#1[#2]{\ensuremath{{#1^{\#}\left[#2\right]}}}
\newcommand\@absstarnounder[1]{%
    \@ifnextchar[{\@absstarbrack{#1}}{\@absstarnothing{#1}}}
\newcommand\@absnostarnounder[1]{%
    \@ifnextchar[{\@absnostarbrack{#1}}{\@absnostarnothing{#1}}}
\def\@absstarunder#1_#2{\ensuremath{\mathrel{#1_{#2}^{\#}}}}
\def\@absnostarunder#1_#2{\ensuremath{{#1_{#2}^{\#}}}}
\newcommand\@absstarchoice[1]{%
    \@ifnextchar_{\@absstarunder{#1}}{\@absstarnounder{#1}}}
\newcommand\@absnostarchoice[1]{%
    \@ifnextchar_{\@absnostarunder{#1}}{\@absnostarnounder{#1}}}
\newcommand\abs{\@ifstar\@absstarchoice\@absnostarchoice}
\makeatother

\newcommand\absint[1]{\llbracket #1 \rrbracket^{\#}}

\newcommand\itvl[2]{[ #1 , #2]}

\newcommand\ptrue{\textit{true}}
\newcommand\pfalse{\textit{false}}
\newcommand\btrue{\ptrue}
\newcommand\bfalse{\pfalse}

\newcommand\eread[2]{#1\funu{#2}}
\newcommand\ewrite[3]{#2\left[#1 \leftarrow #3\right]}

\newcommand\immediateconsequence{\mathcal{H}}

\newcommand\evalr{\ensuremath{\Downarrow}}

\newcommand\extsig[3]{#1\extsymbol #2 \mapsto #3}

\DeclareDocumentCommand{\extsigvec}{ m O{} m O{} m O{n} }{\ifnempty{#1}{#1\extsymbol}\myvecmap{#2}{#3}{\mapsto}{#4}{#5}{#6}}

\newcommand{\While}{\textsc{While}}

\newcommand{\halt}{\bot}
\newcommand{\cont}{\top}
\newcommand{\absflag}{f}

\newcommand\K{\ensuremath{\mathbb{K}}}

\begin{document}

\theoremstyle{definition}
\newtheorem{requirement}[theorem]{Requirement}

\title[Skeletal Semantics]{Skeletal Semantics and Their Interpretations}

\titlenote{This research has been partially supported by the ANR projects AJACS
  ANR-14-CE28-0008 and CISC ANR-17-CE25-0014-01.
  Bodin and Gardner were partially supported by the EPSRC programme
  grant `REMS: Rigorous Engineering of Mainstream Systems', EP/K008528/1.}


\author{Martin Bodin}
\affiliation{\institution{Imperial College London}\country{United Kingdom}}
\author{Philippa Gardner}
\affiliation{\institution{Imperial College London}\country{United Kingdom}}
\author{Thomas Jensen}
\affiliation{\institution{Inria, Univ Rennes, IRISA}\country{France}}
\author{Alan Schmitt}
\affiliation{\institution{Inria, Univ Rennes, IRISA}\country{France}}

\begin{abstract}
  The development of mechanised language specification based on
  structured operational semantics, with applications to verified
  compilers and sound program analysis, requires huge effort.  General
  theory and frameworks have been proposed to help with this
  effort. However, none of this work provides a systematic way of
  developing concrete and abstract semantics, connected together by a
  general consistency result. We introduce a {\em skeletal semantics}
  of a language, where each skeleton describes the complete semantic
  behaviour of a language construct.  We define a general notion of
  {\em interpretation}, which provides a systematic and
  language-independent way of deriving semantic judgements from the
  skeletal semantics.  We explore four generic interpretations: a
  simple well-formedness interpretation; a concrete interpretation; an
  abstract interpretation; and a constraint generator for
  flow-sensitive analysis.  We prove  general {\em consistency
    results} between interpretations, depending only on simple
  language-dependent lemmas. We illustrate our ideas using a simple
  \While{} language. 
\end{abstract}



\maketitle

\section{Introduction}

Plotkin's Structural Operational Semantics~\cite{Plotkin:81:SOS}
provides a methodology for formally describing a programming language
using a collection of inference rules. It has been widely used to
provide, for example, mechanised language specifications of
substantial parts of ML~\cite{Owens:Esop:08},
C~\cite{Norrish-PhD,Blazy-Leroy-Clight-09} and
JavaScript~\cite{BodinChargueraudFilarettiGardnerMaffeisNaudziunieneSchmittSmith2014}. These
specifications have, in turn, been used to build verified
compilers~\cite{2006-Leroy-compcert,CakeML:POPL:14} and to develop
sound program
analysis~\cite{KleinVerifiedJavaBCV,Cachera:05:Extracting,verasco}.
Such language specifications and their applications require huge
effort, stretching the fundamental theory and tools to their 
limits.  Researchers have therefore spent considerable thought
developing general theories and frameworks where some of this effort
can be unified for a wide class of languages.

Abstract interpretation~\cite{Cousot:77:AbstractInterpretation} is a
well-known general theory for analysing programs. It provides general
definitions for describing when an abstract semantics is consistent
(sound) with respect to a concrete semantics, and even suggests a
methodology for how to construct consistent abstract semantics from
concrete semantics~\cite{Cousot:98:Marktoberdorf,
  MidtgaardJ:08,VanHorn:11:Abstracting}.  We focus on abstract
semantics arising from concrete operational semantics. A prominent
example can be found in the Verasco project~\cite{verasco} which
provides a Coq-certified static analyser based on abstract
interpretation, specifically targeting CompCert's mechanised C
specification. Schmidt~\cite{Schmidt:95:Natural, schmidt1997abstract}
has demonstrated how to build   abstract derivations from  concrete
derivations arising from an operational semantics, illustrating a close
connection between the abstract and concrete semantics.
The concepts are general, but the work does not attempt to be systematic.  Inspired
by Schmidt, Bodin {\em et al.}~\cite{cpp2015} have 
identified a general rule format that can be systematically  instantiated to both
concrete and abstract semantics, with a general consistency
result. However, their general rule format is based on a non-standard
style of operational semantics, called {\em pretty-big-step
  operational semantics}~\cite{chargueraud2013pretty}, introduced to
provide a Coq-mechanised specification of
JavaScript~\cite{BodinChargueraudFilarettiGardnerMaffeisNaudziunieneSchmittSmith2014}. It
does not provide a general systematic approach for constructing an
abstract semantics from a standard operational semantics.


A general framework provides a unifying meta-language for writing
operational inference rules, in order to develop general environments
for analysis~\cite{DBLP:conf/lics/HarperHP87,
  DBLP:conf/cade/PfenningS99,
  rosu-2010-jlap,jung2017iris}.
Much of the work on frameworks does not aim to describe abstract
analysis. 
One notable exception is the Iris
framework~\cite{jung2017iris} for reasoning about concurrent
programs. 
Iris provides a systematic
method for building a concurrent program logic from concrete
operational semantics, proving a general consistency result.  It
starts from a concrete operational semantics and generically builds
the program logic. Consequently, the general consistency result relies
on language-dependent lemmas which require an induction over the
possibly complex constructs of the language.   It does not work with
abstract semantics in general, and the lemmas associated with the 
general consistency result are difficult to prove.

%

We introduce a new approach. We have developed a meta-language, which we call a {\em skeletal
  semantics}, from which it is possible to construct systematically
both concrete and abstract semantics,  and prove a
general consistency result.  Our skeletal semantics comprises:
\begin{itemize}
\item {\em
  skeletons}, where each skeleton describes the complete behaviour of
one language construct;
\item generic {\em interpretations}, which systematically derive 
 semantic  judgements from the 
skeletons: for example, a generic concrete interpretation built using 
the usual concrete judgements of an operational semantics, parameterised by an input state,
command and output state, 
and a generic  abstract interpretation built from more abstract
judgements over  abstract domains;
\item a general  {\em consistency result} between interpretations,
 which depends  on simple language-dependent  lemmas. 
\end{itemize}
Our definitions of skeletal semantics and interpretations have been
mechanised in the Coq theorem prover, and the consistency result
proved.

Skeletal semantics can be used to describe languages specified
using big-step operational semantics and languages specified using an
English standard such as the ECMAScript standard. In this introductory
paper, we
focus on a simple \While{} language as the illustrative example; the
lambda calculus
is given in the Coq artefact. Consider the usual $\sif{}{}{}$ command
of a \While{} language, whose behaviour is typically defined in an
operational semantics using two standard rules for the true and false
case. Instead, in our skeletal semantics, the behaviour of the
$\sif{}{}{}$ command is given by one skeleton comprising: a semantic {\em judgement}, in this case
parameterised by input state, expression and value, and instantiated
via the interpretations with, for example, the usual concrete and
abstract judgements for evaluating expressions; then a {\em branch} of
two paths guarded by {\em filters} for determining the true and false
case, followed by judgements for the appropriate subcommands.
Our $\sif{}{}{}$ skeleton thus describes the
information given in the two normal $\sif{}{}{}$ rules, collected
together under one syntactic construct.

Skeletons provide all the information necessary to give systematically
both concrete and abstract interpretations. Intuitively, our generic
concrete interpretation picks one path from each branching/merging of
the skeleton, whereas our generic abstract interpretation
merges
all the appropriate paths. In fact, our interpretations span many
different types of analysis. The paper contains a simple
well-formedness interpretation for simple sorts, suggesting that we
can give many forms of standard well-formedness result associated with
states and types. We also give an interpretation building a
constraint generator for flow-sensitive analysis. 
We discuss
other forms of analyses in future work.

We have proved general consistency results between interpretations,
which depend on simple language-dependent filter lemmas.  These
filter lemmas only describe properties of the filters of a language,
which are functions on the  language values. The complexity of proving
these filter lemmas thus only depends on the complexity the filters,
which are simple in comparison with the complexity of the whole
language. We explore the instantiation of our consistency result for
our \While{} language, demonstrating 
the consistency of the abstract interpretation with respect to the
concrete interpretation  for a selection of domains, as well as the
consistency between the constraint
generation and the abstract interpretation.

In summary, we have come a long way to answering the challenge of
developing a language-independent framework for relating concrete and
abstract semantics. The real test will come  when we move from the
simple languages explored in this paper to real-world languages such
as OCaml and JavaScript, discussed in the future work.

\section*{Example: The \While{} Language}


We demonstrate our skeletal semantics in action using the simple
conditional statement from the \While{} language. 
Consider the usual concrete rules associated with the conditional statement
 in Figure~\ref{fig:if:concrete}, and  the abstract
rules  in Figure~\ref{fig:if:abstract}, 
supposing that the Booleans are abstracted by the usual four-valued lattice
given by \(\mset{\abs\btrue, \abs\bfalse, \top_{\boolT},
  \bot_{\boolT}}\). 
These abstract rules are intuitively correct,
but they are first built in an {ad hoc} way and 
then shown to be related to the concrete rules using a Galois connection.
More generally,
the systematic construction of abstract rules from concrete rules requires a deep understanding
of how the analysed programming language evaluates expressions:
in a case like a vanilla \While{} language, this is quite
straightforward; for a complex language such as JavaScript~\cite{ECMA2018,Maffeis2008An-Operational-Seman,BodinChargueraudFilarettiGardnerMaffeisNaudziunieneSchmittSmith2014},
the relationship  between the concrete and abstract semantics can be difficult
to get right.


\begin{figure}
    \begin{mathpar}
        \inferrule{
            \sigma, e \evalr \btrue \\
            \sigma, t_1 \evalr \ovar}
            {\sigma, \sif{e}{t_1}{t_2} \evalr \ovar}
        \and
        \inferrule{
            \sigma, e \evalr \bfalse \\
            \sigma, t_2 \evalr \ovar}
            {\sigma, \sif{e}{t_1}{t_2} \evalr \ovar}
    \end{mathpar}
    \caption{Usual concrete rules for the \(\sif{}{}{}\) construct}
    \label{fig:if:concrete}
\end{figure}
\begin{figure}
    \begin{mathpar}
        \inferrule{
            \sigma, e \evalr \abs\btrue \\
            \sigma, t_1 \evalr \ovar}
            {\sigma, \sif{e}{t_1}{t_2} \evalr \ovar}
        \and
        \inferrule{
            \sigma, e \evalr \abs\bfalse \\
            \sigma, t_2 \evalr \ovar}
            {\sigma, \sif{e}{t_1}{t_2} \evalr \ovar}
        \and
        \inferrule{
            \sigma, e \evalr \top_{\boolT} \\
            \sigma, t_1 \evalr \ovar \\
            \sigma, t_2 \evalr \ovar}
            {\sigma, \sif{e}{t_1}{t_2} \evalr \ovar}
        \and
        \inferrule{
            \sigma, e \evalr \bot_{\boolT}}
            {\sigma, \sif{e}{t_1}{t_2} \evalr \bot}
    \end{mathpar}
    \caption{Usual abstract rules for the \(\sif{}{}{}\) construct}
    \label{fig:if:abstract}
\end{figure}
\begin{figure}
  \begin{align*}
      \Rulea{If}{\sif{\tvar[1]}{\tvar[2]}{\tvar[3]}}
      {\left[\deriv{\ivar{}}{\tvar[1]}{\fvar[1]};
      \branchesV[\mset{\ovar{}}]{\begin{aligned}
          &\filter{\texttt{isTrue}}{\fvar[1]}{};
          \deriv{\ivar{}}{\tvar[2]}{\ovar{}}\\
          &\filter{\texttt{isFalse}}{\fvar[1]}{};
          \deriv{\ivar{}}{\tvar[3]}{\ovar{}}
        \end{aligned}}
                                   \right]}
  \end{align*}
    \caption{Skeleton for the \(\sif{}{}{}\) construct}
    \label{fig:if:ours}
\end{figure}

We define the skeletal semantics in Section~\ref{sec:rule-format},
which provides a general meta-theory
for defining language semantics.
Figure~\ref{fig:if:ours} shows the skeleton associated with  the
\(\sif{}{}{}\) construct,  with generic subterms denoted by 
 \(\tvar[1]\), \(\tvar[2]\),
and \(\tvar[3]\),  input state \(\ivar\) and output state \(\ovar\).
Judgements of the form $\deriv{-}{\tvar[i]}{-}$ identify the
required subcomputations associated with  the subterms $\tvar[1], \tvar[2], \tvar[3]$.
 The skeleton stitches these judgements  together,
using the input and output states, the internal  symbolic variable
$\fvar[1]$, and  the branching 
which identifies  paths through the skeleton  using the
filters \(\mathtt{isTrue}\) and \(\mathtt{isFalse}\), resulting in the
output state  \(\ovar\).
Such a skeleton thus explicitly describes both the data flow and the control flow
associated with a language construct,
identifying the common pattern  underlying the concrete and abstract rules.

%
%

%

We provide a general definition of interpretation for our skeletal semantics
in Section~\ref{SEC:INTERP} and study four generic interpretations:
\begin{itemize}
    \item A simple well-formedness interpretation
      (Section~\ref{SEC:INTERP}), which states that
      the stitching of the skeleton in Figure~\ref{fig:if:ours} respects the sorting of
      the basic constructs. 
    \item The concrete interpretation
      (Section~\ref{SEC:CONCRETE-INTERPRETATION}),
which intuitively picks one path from each branching of the skeleton, corresponding to
the two rules of Figure~\ref{fig:if:concrete}.
    \item The abstract interpretation (Section~\ref{SEC:ABSTRACT-INTERPRETATION}),
        whose complex definition 
(Figure~\ref{fig:abstract-interpretation})
  boils down to the intuitive description given by the  rule of Figure~\ref{fig:if:abstract:ours}:
        a rule with optional branches,
        considering all paths compatible with
        the return value of the expression \(e\).
        This rule naturally subsumes the four rules of Figure~\ref{fig:if:abstract}.
    \item A constraint generator for flow-sensitive static analysis
      (Section~\ref{SEC:CONSTRAINTS}).
        Although these constraints are different in nature to the abstract semantics,
        they are expressed in our meta-theory using the same
        mechanism: that is, 
        an interpretation of the skeletal semantics.
        This provides a strong connection between them.
\end{itemize}
We also provide general definitions  of \emph{consistency} between
interpretations (Section~\ref{sec:consistency}), with  general
consistency proofs based on filter lemmas
(Section~\ref{subsec:proving-consistency}).
The shared structure of our different interpretations greatly eases the proof process.
We use our consistency definitions  to show that the abstract interpretation is correct
with respect to the concrete interpretation,
and that any solution to the constraints given by our constraint
generator must give rise to a correct
abstract semantics.

Throughout the paper, we  instantiate  our definitions and results to
the \While{} language as a way of  introducing our ideas and 
demonstrating  how classic proof techniques based on an abstract interpretation of
\While{} can be captured with our approach (Section~\ref{SEC:PROOF-TECHNIQUES}).
We however emphasise that skeletons and interpretations,
as well as their consistency proofs,
are generic and can be applied to any programming language. To begin
to illustrate this,
we extend our \While\; language with exceptions, input/output and a heap in
Section~\ref{sec:while2}.

The definitions and proofs of Sections \ref{sec:rule-format} to
\ref{SEC:ABSTRACT-INTERPRETATION} have been formalised in Coq; those
of  Sections~\ref{SEC:PROOF-TECHNIQUES}
and~\ref{SEC:CONSTRAINTS} have been proven on paper.  They are all
available from the companion
website\footnote{\url{http://skeletons.inria.fr}}.

\begin{figure}
    \begin{mathpar}
        \inferrule{
            \astate, e \evalr \aval \\
            \mpar{\absint{\texttt{isTrue}}\funu{\aval} \implies \astate, t_1 \evalr \astate[o]} \\
            \mpar{\absint{\texttt{isFalse}}\funu{\aval} \implies \astate, t_2 \evalr \astate[o]}}
        {\astate, \sif{e}{t_1}{t_2} \evalr \astate[o]}
    \end{mathpar}
    \caption{The abstract interpretation
      of the \(\sif{}{}{}\) construct: intuitive description.}
    \label{fig:if:abstract:ours}
\end{figure}

\section{Skeletal Semantics}\label{sec:rule-format}

\subsection{Terms}\label{subsec:terms}

Terms \(t\) of a skeletal semantics are built using base terms, term variables,
and constructors. Base terms are left unspecified and correspond to the basic
blocks of the syntax, such as literals or program identifiers. They are
instantiated by interpretations. We assume a countable set of term variables,
ranged over by \(\tvar{}\), and a finite set of constructors, ranged over by
\(c\). A term is thus a base term, a term variable, or a constructor applied to
terms.

We also assume a countable set of {\em sorts}, ranged over by \(s\). The sorts
are separated into \emph{base sorts}, for base terms, and \emph{program sorts},
for terms built using constructors. Any base term belongs to a single base sort.
The \emph{signature} of a constructor \(c\), written \(\sig{c}\), is of the form
\(\myvecp{\sort}{n} \rightarrow \sort\), where \(n\) is the arity of \(c\), the
\(\sort_{i}\) for $i = 1..n$ are sorts, and \(\sort\) is a program sort.

\paragraph{Running Example} For the  \While{} language, the base
sorts are \(\identS\) for the program
  variables and \(\litS\) for the literals. Program sorts are \(\exprS\) for
expressions  and \(\statS\) for statements. The signature of
constructors is given in Figure~\ref{fig:example:constructors}.

\begin{figure}
  \centering
  \begin{equation*}
    \begin{array}[t]{|c|c|}
      \hline
      \text{\(c\)} & \text{Signature} \\
      \hline
      \constn & \litS \rightarrow \exprS \\
      \varn & \identS \rightarrow \exprS\\
      + & (\exprS \times \exprS) \rightarrow \exprS\\
      = &  (\exprS \times \exprS) \rightarrow \exprS\\
      \neg & \exprS \rightarrow \exprS\\
      \hline
    \end{array}
    \quad
    \begin{array}[t]{|c|c|}
      \hline
      \text{\(c\)} & \text{Signature} \\
      \hline
      \sskip & \statS\\
      :=  & (\identS \times \exprS) \rightarrow \statS\\
      ; & (\statS \times \statS) \rightarrow \statS\\
      \iname & (\exprS \times \statS \times \statS) \rightarrow \statS\\
      \wname & (\exprS \times \statS) \rightarrow \statS\\
      \hline
    \end{array}
  \end{equation*}
  \caption{Constructors for \While{}}
  \label{fig:example:constructors}
\end{figure}

Let \(\typenv\) be a mapping from term variables to sorts. Sorted terms are
either base terms, term variables \(\tvar{}\) of sort \(\typenv\funu{\tvar{}}\), or a term
$c\myvecip{t}{n}$ of sort \(\sort\), where \(c\) has signature \(\sig{c} =
\myvecp{\sort}{n} \rightarrow \sort\) and the terms \(\myveci{t}{n}\) have
the appropriate sort. We write \(\Sort{t}\) for the sort of \(t\). Let
\(\termenv\) be a mapping from term variables to terms such that \(\forall
\tvar{} \in \dom{\termenv}, \Sort{\termenv\funu{\tvar{}}} =
\typenv\funu{\tvar{}}\). We extend it to terms as \(\termenv\funu{c
  \myvecip{t}{n}} = c \myvecip[\termenv]{t}{n}\) when defined. We write
\(\Sort[]{t}\) for \(\Sort[\emptyset]{t}\) and \(\Tvar{t}\) for the set of term variables in \(t\).
We say \(t\) is \emph{closed} if \(\Tvar{t} = \emptyset\). In that case, we
write \(t : \sort\) for \(\Sort[]{t} = \sort\).

\begin{lemma}\label{lem:sort-termenv}
  Let \(\term\) a term, \(\termenv\) an environment mapping term variables to
  closed terms, and \(\typenv\) a sorting environment such that \(\Tvar{\term}
  \subseteq \dom{\termenv}\), \(\Tvar{\term} \subseteq \dom{\typenv}\), and for
  any \(\tvar \in \Tvar{\term}\) we have \(\typenv\funu{\tvar} =
  \Sort[]{\termenv\funu{\tvar}}\). Then we have \(\Sort{\term} =
  \Sort[]{\termenv\funu{\term}}\).
\end{lemma}

\subsection{Skeletons}\label{subsec:skeletons}

We assume a countable set of \emph{flow variables}, ranged over by \(\fvar\),
which are used in the skeleton bodies to hold semantic values (states,
intermediate values, \ldots). Among flow variables, we distinguish two of them:
\(\ivar\) holds the semantic state at the start of a skeleton, and
\(\ovar\) is supposed to hold the semantic result at the end of a skeleton.
We let \emph{skeletal variables}, ranged over by \(x\) or \(y\), be the union of
term variables and flow variables. A \emph{skeleton} has the shape
\(\Rule{Name}{c\myvecp{\tvar}{n}}{\skel}\), where \(\rulen{Name}\) is the
skeleton name, \(c\) is a constructor, \(\myvec{\tvar}{n}\) are term variables, and \(\skel\) is the \emph{skeleton body}:
\begin{align*}
\textsc{Skeleton Body} \quad  \skel & \Coloneqq \emptylist \mid \cons{B}{\skel}\\
  \textsc{Bone} \quad  B & \Coloneqq \deriv{x_{f_1}}{t}{x_{f_2}} \mid \filter{F}{\myvec{\svar{x}}{n}}{\myvecp{\svar{y}}{m}}
      \mid \branches{\myvec{\skel}{n}}
\end{align*}
where $\deriv{-}{-}{-}$ is the (terminal) hook constructor and
$\filter{F}{}{}$ ranges over the set of filter functions. 

A skeleton body is a sequence of \emph{bones}. A bone is either a
\emph{hook judgement}
\(\deriv{\fvar[1]}{t}{\fvar[2]}\), built using the constructor
$\deriv{-}{-}{-}$ from an input flow variable
\(\fvar[1]\), a term \(t\) to be hooked during interpretation, and an output
flow variable \(\fvar[2]\); or a \emph{filter}
\(\filter{F}{\myvec{\svar{x}}{n}}{\myvecp{\svar{y}}{m}}\) which tests if the
values bound to its input skeletal variables \(\myvecp{\svar{x}}{n}\) satisfy a
condition specified by \(F\), and in that case outputs values to be bound to
\(\myvecp{\svar{y}}{m}\); or a set of \emph{branches}
\(\branches{\myvec{\skel}{n}}\) which represent the different behavioural
pathways, where \(\bvar\) declares the skeletal variables that are shared and
must be defined by all branches.

A filter with no output skeletal variables is simply written
\(\filter{F}{\myvec{\svar{x}}{n}}{}\). It then acts as a predicate.

\begin{requirement}\label{req:single-skeleton}
  We require that there exists exactly one skeleton
  for any given constructor~\(c\).
\end{requirement}

\paragraph{Running Example}
The skeletons of our \While{} example are given in
Figure~\ref{fig:example:concrete:semantics:translated}.
Requirement~\ref{req:single-skeleton} is trivially satisfied.

\begin{figure}
  \begin{align*}
      \Rulea{Lit}{\const{\tvar}}{\left[\filter{\texttt{litInt}}{\tvar}{\fvar[1]};
    \filter{\texttt{intVal}}{\fvar[1]}{\ovar}\right]}
    \\
      \Rulea{Var}{\var{\tvar}}
      {\left[\filter{\texttt{read}}{\tvar{},\ivar{}}{\ovar{}}\right]}
    \\
      \Rulea{Add}{\tvar[1]{} + \tvar[2]{}}
      {\left[
    \begin{multlined}[][\arraycolsep]
        \deriv{\ivar{}}{\tvar[1]{}}{\fvar[1]{}};
        \filter{\texttt{isInt}}{\fvar[1]{}}{\fvar[1']};
        \deriv{\ivar{}}{\tvar[2]{}}{\fvar[2]{}};\\
        \filter{\texttt{isInt}}{\fvar[2]{}}{\fvar[2']{}};
        \filter{\texttt{add}}{\fvar[1']{}, \fvar[2']{}}{\fvar[3]{}};
        \filter{\texttt{intVal}}{\fvar[3]}{\ovar{}}
      \end{multlined}
      \right]}
    \\
      \Rulea{Eq}{\tvar[1]{} = \tvar[2]{}}
      {\left[
    \begin{multlined}[][\arraycolsep]
        \deriv{\ivar{}}{\tvar[1]{}}{\fvar[1]{}};
        \filter{\texttt{isInt}}{\fvar[1]}{\fvar[1']};
        \deriv{\ivar{}}{\tvar[2]{}}{\fvar[2]{}};\\
        \filter{\texttt{isInt}}{\fvar[2]}{\fvar[2']};
        \filter{\texttt{eq}}{\fvar[1']{}, \fvar[2']{}}{\fvar[3]{}};
        \filter{\texttt{boolVal}}{\fvar[3]}{\ovar{}}
      \end{multlined}
      \right]}
    \\
      \Rulea{Neg}{\neg{\tvar}}
      {\left[ \deriv{\ivar{}}{\tvar}{\fvar[1]};
      \filter{\texttt{isBool}}{\fvar[1]}{\fvar[2]};
      \filter{\texttt{neg}}{\fvar[2]}{\fvar[3]{}};
        \filter{\texttt{boolVal}}{\fvar[3]}{\ovar{}}
    \right]}
    \\
      \Rulea{Skip}{\sskip}
      {\left[ \filter{\texttt{id}}{\ivar{}}{\ovar{}} \right]}
    \\
      \Rulea{Asn}{\asnescape{\tvar[1]}{\tvar[2]}}
      {\left[ \deriv{\ivar{}}{\tvar[2]}{\fvar[1]};
      \filter{\texttt{write}}{\tvar[1]{},\ivar{},\fvar[1]}{\ovar{}} \right]}
    \\
      \Rulea{Seq}{\seq{\tvar[1]}{\tvar[2]}}
      {\left[ \deriv{\ivar{}}{\tvar[1]}{\fvar[1]};
      \deriv{\fvar[1]}{\tvar[2]}{\ovar{}} \right]}
    \\
      \Rulea{If}{\sif{\tvar[1]}{\tvar[2]}{\tvar[3]}}
      {\left[ \deriv{\ivar{}}{\tvar[1]}{\fvar[1]};
      \filter{\texttt{isBool}}{\fvar[1]}{\fvar[1']};
      \branchesV[\mset{\ovar{}}]{\begin{aligned}
          &\filter{\texttt{isTrue}}{\fvar[1']}{};
          \deriv{\ivar{}}{\tvar[2]}{\ovar{}}\\
          &\filter{\texttt{isFalse}}{\fvar[1']}{};
          \deriv{\ivar{}}{\tvar[3]}{\ovar{}}
        \end{aligned}}
                                   \right]}
    \\
      \Rulea{While}{\while{\tvar[1]}{\tvar[2]}}
      {\left[
      \begin{multlined}[][\arraycolsep]
        \deriv{\ivar{}}{\tvar[1]}{\fvar[1]};
        \filter{\texttt{isBool}}{\fvar[1]}{\fvar[1']};\\
        \quad\branchesV[\mset{\ovar{}}]{\begin{aligned}
                &\filter{\texttt{isTrue}}{\fvar[1']}{};
                \deriv{\ivar{}}{\tvar[2]}{\fvar[2]};
                \deriv{\fvar[2]}{\while{\tvar[1]}{\tvar[2]}}{\ovar{}}\\
                &\filter{\texttt{isFalse}}{\fvar[1']}{};
                \filter{\texttt{id}}{\ivar{}}{\ovar{}}
          \end{aligned}
        }
      \end{multlined}
                                   \right]}
  \end{align*}
  \caption{Skeletal semantics for \While{}}\label{fig:example:concrete:semantics:translated}
\end{figure}

\subsection{Flow Sorts}\label{subsec:flow-sorts}

We extend the sorts with \emph{flow sorts}, that are the sorts of values in
interpretations. In our running example, flow sorts are \(\stateT\) for the
variable store, \(\valT\) for values, \(\intT\) for integers, and \(\boolT\) for
Booleans. We relate flow sorts to hooks and filters as follows.

In a hook \(\deriv{\fvar[1]}{t}{\fvar[2]}\), the flow variable \(\fvar[1]\)
stands for an input state that fits with \(t\), and \(\fvar[2]\) stands for a
result. Given a program sort \(s\), we define \(\typin{s}\) as its input flow
sort and \(\typout{s}\) as its output flow sort. In our running example, the
input flow sort of both expressions and statements is \(\stateT\). The output
flow sort of expressions is \(\valT\) and the output flow sort of statements is
\(\stateT\).

Similarly, a filter \(\filter{F}{\myvec{\svar{x}}{n}}{\myvecp{\svar{y}}{m}}\) is
assigned a signature, written \(\typsig{F}\), of the form \(\myvecp{\sort}{n}
\rightarrow \myvecp{\sortp}{m}\). We write \(\unit\) for the output sort of a
filter if \(m = 0\) and omit the enclosing parentheses when \(n\) or \(m\) is
\(1\). Filter signatures for our running example are given in
Figure~\ref{fig:example:typing:filters}.

\begin{figure}
  \centering
\begin{equation*}
\begin{array}[t]{|c|c|}
  \hline
  f & \typsig{f}\\
  \hline
  \texttt{litInt} & \litS \rightarrow \intT\\
  \texttt{intVal} & \intT \rightarrow \valT\\
  \texttt{isInt} & \valT \rightarrow \intT\\
  \texttt{add} & \mpar{\intT,\intT} \rightarrow \intT\\
  \hline
\end{array}
\quad
\begin{array}[t]{|c|c|}
  \hline
  f & \typsig{f}\\
  \hline
  \texttt{boolVal} & \boolT \rightarrow \valT\\
  \texttt{isBool} & \valT \rightarrow \boolT\\
  \texttt{isTrue} & \boolT \rightarrow \unit \\
  \texttt{isFalse} & \boolT \rightarrow \unit \\
  \hline
\end{array}
\quad
\begin{array}[t]{|c|c|}
  \hline
  f & \typsig{f}\\
  \hline
  \texttt{eq} & \mpar{\intT,\intT} \rightarrow \boolT\\
  \texttt{neg} & \boolT \rightarrow \boolT\\
  \texttt{read} & \mpar{\identS,\stateT} \rightarrow \valT\\
  \texttt{write} & \mpar{\identS, \stateT, \valT} \rightarrow \stateT\\
  \texttt{id} & \stateT \rightarrow \stateT\\
  \hline
\end{array}
\end{equation*}
  \caption{Filter sorts}\label{fig:example:typing:filters}
\end{figure}

We check the consistency of the hook and filters with the skeletons in our
well-formedness interpretation, introduced in
Section~\ref{subsec:wf-interpretation}.

\section{Interpretations}\label{SEC:INTERP}

An \emph{interpretation} \(I\) specifies base terms and how to interpret the empty skeleton body,
hooks, filters, and branches. It defines a set of \emph{interpretation states},
ranged over by \(\interpIn\) in this section but with specific notations for
each interpretation, and a set of \emph{interpretation results}, ranged over by
\(\interpOut\) in this section, as well as the following relations:
\begin{itemize}
  \item \(\interp{\emptylist}\) defining the interpretation of the empty
    skeleton body;
  \item \(\interp{\deriv{\fvar[1]}{t}{\fvar[2]}}[\interpIn']\) defining the interpretation of a hook;
  \item \(\interp{\filter{F}{\myvec{\svar{x}}{n}}{\myvecp{\svar{y}}{m}}}[\interpIn']\) defining the
    interpretation of a filter, for each filter \(F\);
  \item \(\merge{\outfun,\interpIn}{n}\) defining the merging of the interpretation of
    branches, where \(\outfun\) is a partial function from \([1..n]\) to
    interpretation results, and where \(\bvar\) is the set of skeletal variables
    defined and shared by all branches.
\end{itemize}

Given a skeleton body \(S\) and the relations above, we define the remaining
cases for the interpretation of \(S\) as follows.
\begin{align*}
  \mpar{
  \begin{gathered}
    \interp{B}[\interpIn']\\
    \interp[\interpIn']{S}
  \end{gathered}
  }
  &\implies
    \interp{\cons{B}{S}}\\
  \mpar{
  \begin{gathered}
    \forall i \in \dom{\outfun}.\interp{S_i}[\outfun\funu{i}]\\
    \merge{\outfun,\interpIn}{n}
  \end{gathered}
}
  &\implies
    \interpb{\branches{\myvec{\skel}{n}}}[\interpIn']
\end{align*}

Interpretations enable us to define the meaning of skeletons
by only specifying the parts that matter.
Interpretations apply to any skeletons and are thus independent of the language.
The rest of the paper presents different interpretation
and their relations.

\subsection{Well-Formedness Interpretation}\label{subsec:wf-interpretation}

The first interpretation we consider is a \emph{well-formedness} interpretation,
to verify that every skeleton is well formed. More precisely, we verify that
every skeletal variable used has been first defined, that every variable defined
in a skeleton is fresh (with an exception for branches, see below), and that the
sorting of filters, hooks, and branches are consistent.

Intuitively, in the hook \(\deriv{\fvar[1]}{t}{\fvar[2]}\), flow variable
\(\fvar[1]\) is \emph{used} and flow variable \(\fvar[2]\) is \emph{defined}.
Similarly, in the filter
\(\filter{F}{\myvec{\svar{x}}{n}}{\myvecp{\svar{y}}{m}}\), skeletal variables
\(\myvecp{\svar{x}}{n}\) are used and skeletal variables
\(\myvecp{\svar{y}}{m}\) are defined.
The case for branches \(\branches{\myvec{\skel}{n}}\) is a bit more involved. First, each
branch \(\skel[i]\) \emph{must} define the skeletal variables in \(\bvar\).
Second, every variable defined in the whole set of branches must be distinct,
with the exception of the variables in \(\bvar\) as they have to be defined in
every branch. And third, the only variables defined by the branches that may be
used in the rest of the skeleton body are those in \(\bvar\).

Assuming for each base sort a set of base terms, pairwise disjoint,
we define the well-formedness (WF) interpretation in Figure~\ref{fig:typing-interpretation}.
Its interpretation states and result consist of a pair of a sorting environments
\(\typenv\), mapping term variables to base and program sorts, and flow variables
to flow sorts, and a set \(\dfvarset\) of skeletal variables that have been defined at
that point. In this interpretation, we write \(\svar{x} : \sort\) to state that
the kind of variable and sort match, namely term variables with base or program
sorts, and flow variables with flow sorts.

\begin{figure}
  \centering
  \begin{align*}
    &\implies \wfinterp{\emptylist}[\mpar{\typenv,\dfvarset}]\\
    \mpar
    {
    \begin{rgathered}
      \fvar[1] \in \dom{\typenv} \subseteq \dfvarset\\
      \Tvar{\term} \subseteq \dom{\typenv}\\
      \typenv\funu{\fvar[1]} = \typin{\Sort{t}}\\
      \fvar[2] \notin \dfvarset\\       
      \typenv' = \extenv{\typenv}{\fvar[2]}{\typout{\Sort{t}}}\\
      \dfvarset' = \dfvarset \cup \mset{\fvar[2]}
    \end{rgathered}
    }
    &\implies \wfinterp{\deriv{\fvar[1]}{t}{\fvar[2]}}\\
    \mpar{
    \begin{rgathered}
      \myvecp{\svar{x}}{n} \subseteq \dom{\typenv} \subseteq \dfvarset\\
      \myvecp{\svar{x}}{n} : \myvecp[\typenv]{\svar{x}}{n}\\
        \myvec{\svar{y}}{m} \text{ are pairwise distinct}\\
      \myvecp{\svar{y}}{m} \cap \dfvarset = \emptyset\\
      \typsig{F} = \myvecp[\typenv]{\svar{x}}{n} \rightarrow \myvecp{\sort}{m}\\
      \myvecp{\svar{y}}{m} : \myvecp{\sort}{m}\\
      \typenv' = \extsig{\typenv}{\myvecp{\svar{y}}{m}}{\myvecp{\sort}{m}}\\
      \dfvarset' = \dfvarset \cup \myvecp{\svar{y}}{m}
    \end{rgathered}}
    &\implies \wfinterp{\filter{F}{\myvec{\svar{x}}{n}}{\myvecp{\svar{y}}{m}}}\\
    \mpar{
    \begin{rgathered}
      n \geq 2\\
      \dom{\typenv} \subseteq \dfvarset\\
      \forall i \in [1..n]. \outfun\funu{i} = \mpar{\typenv_i,\dfvarset_i}\\
      \forall i \in [1..n]. \dom{\typenv_i} \subseteq \dfvarset_i\\
      \forall i j. i \neq j \implies \mpar{\dfvarset_i \setminus \dfvarset} \cap
      \mpar{\dfvarset_j \setminus \dfvarset} = \bvar\\
      \forall i \in [1..n]. \typenv \extsymbol \envrestr{\typenv_i} = \typenv'\\
      \dfvarset' = \bigcup_{i \in [1..n]} \dfvarset_i
    \end{rgathered}}
    &\implies \wfinterp[\outfun,{\mpar{\typenv,\dfvarset}}]{\bigoplus_n}[\mpar{\typenv',\dfvarset'}][\bvar]
  \end{align*}
  \caption{WF Interpretation}\label{fig:typing-interpretation}
\end{figure}

The interpretation for the empty skeleton body is trivial, it simply returns its
arguments. The interpretation of a hook \(\deriv{\fvar[1]}{t}{\fvar[2]}\) checks
that \(\fvar[1]\) is in \(\typenv\), that every term variable of \(t\) is also
in \(\typenv\), and that variable \(\fvar[2]\) is fresh (i.e., not in
\(\dfvarset\)). In addition, it checks that the sort for \(\fvar[1]\) is what
\(t\) expects as input sort and that \(\fvar[2]\) is latter bound to an
output sort of \(t\). 

The interpretation for a filter
\(\filter{F}{\myvec{\svar{x}}{n}}{\myvecp{\svar{y}}{m}}\) is similar. It ensures
that the input skeletal variables \(\myvecp{\svar{x}}{n}\) are in \(\typenv\),
that the number and kind of both input and output variables match the signature
of \(F\), that the output variables are fresh, that the sort of the input
variable corresponds to the input signature of \(F\), and it continues binding
the output variables to the output signature of \(F\).

Finally, the interpretation of the merging of branches checks that every branch
is well formed, that the variables in \(\bvar\) are exactly those shared by the
branches (neither less nor more than those), and that the sorting environments
returned by the branches all agree when restricted to \(\bvar\). In that case, the
returned sorting environment is the concatenation of the input environment and
the one shared by the branches. The \(n \geq 2\) constraint is to have a more
concise way of stating that the variables shared by the branches are exactly
those in \(V\). It is not a restriction as an empty set of branches is useless,
it prevents the skeleton from being interpreted as offering no pathway, and a
singleton set of branches can be inlined.

Let \(t = c\myvecp{\term}{n}\) be a closed term such that \(\Sort[]{t} =
\sort\), where \(\sort\) is a program sort. There are two ways to assign an
output sort to \(t\): directly, as \(\typout{\sort}\), or using the
WF interpretation of the skeleton for \(c\) to compute the
associated sort \(\ovar\). If both coincide, we say the skeleton is \emph{well
  formed}.

\begin{definition}
  A skeleton \(\Rule{Name}{c\myvecp{\tvar}{n}}{\skel}\) is \emph{well formed} iff
  for any closed term \(t = c\myvecp{\term}{n}\) such that \(\Sort[]{t} =
  \sort\), we have \(\wfinterp{\skel}\) and \(\typenv'\funu{\ovar} =
  \typout{\sort}\), with the initial sorting environment \(\typenv\) being
  \(\mset{\extsigvec{\ivar \mapsto \typin{\sort}}{\tvar}[\Sort[]{}]{\term}}\),
  and with \(\dfvarset = \dom{\typenv}\).
\end{definition}

In the following we only consider well-formed skeletons. For instance, the skeletons for \While{} are well formed.

\subsection{Interpretation Consistency}\label{sec:consistency}

We now define how to relate interpretations. Given interpretations \(I_1\) and
\(I_2\), we assume a relation \(\OKst{\interpIn_1}{\interpIn_2}\) between
the interpretation states, and a relation \(\OKout{\interpOut_1}{\interpOut_2}\)
between their results.
Intuitively,
consistency is the propagation of these relations along interpretations.

We define two kinds of consistency: one about where interpretations are defined,
i.e, whether they return a result, and one about their results.

\begin{definition}
  Interpretation \(I_1\) is \emph{existentially consistent} with interpretation
  \(I_2\) if for
  any \(\skel\), \(\interpIn_1\), \(\interpIn_2\), and \(\interpOut_1\), such
  that \(\OKst{\interpIn_1}{\interpIn_2}\) and
  \(\interp[\interpIn_1][I_1]{\skel}[\interpOut_1]\), there exists a
  \(\interpOut_2\) such that \(\interp[\interpIn_2][I_2]{\skel}[\interpOut_2]\)
  and \(\OKout{\interpOut_1}{\interpOut_2}\).
\end{definition}

\begin{definition}
  Interpretations \(I_1\) and \(I_2\) are \emph{universally consistent} if for
  any \(\skel\), \(\interpIn_1\), \(\interpIn_2\), \(\interpOut_1\), and
  \(\interpOut_2\), if \(\OKst{\interpIn_1}{\interpIn_2}\),
  \(\interp[\interpIn_1][I_1]{\skel}[\interpOut_1]\) and
  \(\interp[\interpIn_2][I_2]{\skel}[\interpOut_2]\), then
  \(\OKout{\interpOut_1}{\interpOut_2}\).
\end{definition}

\subsection{Proving Consistency}\label{subsec:proving-consistency}

Both consistency properties can be stated at the level of the building block of
interpretations. Formally, we have the following two lemmas.

\begin{lemma}\label{lem:dc}
    Let \(I_1\) and \(I_2\) be two interpretations, \(\OKst{}{}\) a relation
  between their input states, and \(OKout{}{}\) a relation between their output
  states. If for any \(\interpIn_1\) and \(\interpIn_2\) such that
  \(\OKst{\interpIn_1}{\interpIn_2}\) we have
  \begin{enumerate}
    \item\label{dc-nil} \(\interp[\interpIn_1][I_1]{\emptylist}[\interpOut_1]
      \implies \exists O_2.\;\interp[\interpIn_2][I_2]{\emptylist}[\interpOut_2] \;\land\;
      \OKout{\interpOut_1}{\interpOut_2}\)
    \item\label{dc-hook}
      \(\interp[\interpIn_1][I_1]{\deriv{\fvar[1]}{t}{\fvar[2]}}[\interpIn'_1]
      \implies \exists \interpIn'_2.\;
      \interp[\interpIn_2][I_2]{\deriv{\fvar[1]}{t}{\fvar[2]}}[\interpIn'_2] \;\land\;
      \OKst{\interpIn'_1}{\interpIn'_2}\)
    \item\label{dc-filter}
      \(
      \begin{multlined}[t][\arraycolsep]
        \textstyle
        \interp[\interpIn_1][I_1]{\filter{F}{\myvec{\svar{x}}{n}}{\myvecp{\svar{y}}{m}}}[\interpIn'_1]
        \implies\\
        \textstyle \exists \interpIn'_2.\;
        \interp[\interpIn_2][I_2]{\filter{F}{\myvec{\svar{x}}{n}}{\myvecp{\svar{y}}{m}}}[\interpIn'_2]
        \;\land\; \OKst{\interpIn'_1}{\interpIn'_2}
      \end{multlined}
      \)
    \item\label{dc-merge} \(
      \begin{multlined}[t][\arraycolsep]
        \textstyle
        \dom{\outfun_1} = \dom{\outfun_2} \subseteq \mset{1..n} \;\land\;
        \forall i \in \dom{\outfun_1}.
        \OKout{\outfun_1\funu{i}}{\outfun_2\funu{i}} \\
        \textstyle \land\;
        \merge{\outfun_1,\interpIn_1}{n}[\interpIn'_1] \implies \exists
        \interpIn'_2.\; \merge{\outfun_2,\interpIn_2}{n}[\interpIn'_2] \land
        \OKst{\interpIn'_1}{\interpIn'_2}
      \end{multlined}
\)
  \end{enumerate}
  then \(I_1\) is existentially consistent with \(I_2\).


\end{lemma}

\begin{lemma}\label{lem:rc}
    Let \(I_1\) and \(I_2\) be two interpretations, and \(\OKst{}{}\) a relation
  between their input states. If for any \(\interpIn_1\) and \(\interpIn_2\)
  such that \(\OKst{\interpIn_1}{\interpIn_2}\) we have
  \begin{enumerate}
    \item\label{rc-nil}
      \(\interp[\interpIn_1][I_1]{\emptylist}[\interpOut_1] \; \land \;
      \interp[\interpIn_2][I_2]{\emptylist}[\interpOut_2] \implies
      \OKout{\interpOut_1}{\interpOut_2}\)
    \item\label{rc-hook}
      \(\interp[\interpIn_1][I_1]{\deriv{\fvar[1]}{t}{\fvar[2]}}[\interpIn'_1]
      \; \land \;
      \interp[\interpIn_2][I_2]{\deriv{\fvar[1]}{t}{\fvar[2]}}[\interpIn'_2] \implies
      \OKst{\interpIn'_1}{\interpIn'_2}\)
    \item\label{rc-filter} 
      \(\interp[\interpIn_1][I_1]{\filter{F}{\myvec{\svar{x}}{n}}{\myvecp{\svar{y}}{m}}}[\interpIn'_1]
      \; \land \;
      \interp[\interpIn_2][I_2]{\filter{F}{\myvec{\svar{x}}{n}}{\myvecp{\svar{y}}{m}}}[\interpIn'_2]
      \implies \OKst{\interpIn'_1}{\interpIn'_2}\)
    \item\label{rc-merge} \(
      \begin{multlined}[t][\arraycolsep]
        \textstyle
        \dom{\outfun_1} \subseteq \mset{1..n} \; \land \dom{\outfun_2} \subseteq
        \mset{1..n} \; \land \; \forall i \in \dom{\outfun_1} \cap
          \dom{\outfun_2}. \OKout{\outfun_1\funu{i}}{\outfun_2\funu{i}} \\
        \textstyle
        \land \; \merge{\outfun_1,\interpIn_1}{n}[\interpIn'_1]
        \;\land \;
        \merge{\outfun_2,\interpIn_2}{n}[\interpIn'_2] \implies
        \OKst{\interpIn'_1}{\interpIn'_2}
      \end{multlined}
\)
  \end{enumerate}
  then \(I_1\) and \(I_2\) are universally consistent.


\end{lemma}

\section{Concrete Interpretation}\label{SEC:CONCRETE-INTERPRETATION}

We now define an interpretation used to compute a big-step evaluation semantics
in the form of a \emph{triple set}: a set of triples (also called \emph{judgements})
of the form \emph{(state, term, result)}.
For each base sort we assume a set of base terms, pairwise disjoint,
and for each flow sort a set of \emph{values}. We write \(\term : \sort\) to
state that base term \(\term\) has base sort \(\sort\), and \(\cval : \sort\) to
state that value \(\cval\) has flow sort \(\sort\).

For each filter \(\filter{F}{\myvec{\svar{x}}{n}}{\myvecp{\svar{y}}{m}}\) such
that \(\typsig{F} = \myvecp{\sort}{n} \rightarrow \myvecp{\sortp}{m}\), we
assume an interpretation \(\concrfilter{F}{}{}\) which is a relation
between elements of \(\myvecp{\sort}{n}\) and elements of
\(\myvecp{\sortp}{m}\). We write
\(\concrfilter{F}{\myvec{\cval}{n}}{\myvecp{\cvalp}{m}}\) to state it relates
\(\myvecp{\cval}{n}\) to \(\myvecp{\cvalp}{m}\).

The input state of a concrete interpretation is a pair comprising
\begin{itemize}
  \item an environment \(\concrenv\) mapping term variables to closed terms and
    flow variables to values,
  \item a set \(\tripleset\) of triples of value, closed term, and value,
    representing already known judgements and used to give meaning to the
    sub-derivations $\deriv{\fvar[1]}{t}{\fvar[2]}$.
\end{itemize}
The interpretation result maps term variables to closed terms and flow
variables to values. 

We define the concrete interpretation in
Figure~\ref{fig:concrete-interpretation}. For the empty skeleton body, it simply
returns its environment. For a hook \(\deriv{\fvar[1]}{\term}{\fvar[2]}\), it
looks up in the triple set a known computation for \(\concrenv\funu{\fvar[1]}\)
and \(\concrenv\funu{\term}\) whose result is \(\cval\), and it continues
binding \(\fvar[2]\) to \(\cval\). Note that if the language is
non-deterministic, there may be several such values and one is picked. For a
filter \(F\), one uses its interpretation with the input
\(\myvecp[\concrenv]{\svar{x}}{n}\). As filter interpretations are relations,
there may be several results as well. Finally, to merge branches, the
interpretation picks a branch that successfully returned a result and extends
its environment accordingly.

\begin{figure}
  \centering
  \begin{align*}
    &\implies \concrinterp{\emptylist}[\concrenv]\\
    \mpar{\begin{rgathered}
        \mpar{\concrenv\funu{\fvar[1]},\concrenv\funu{\term},\cval} \in \tripleset\\
        \concrenv' = \extenv{\concrenv}{\fvar[2]}{\cval}
      \end{rgathered}}
    &\implies \concrinterp{\deriv{\fvar[1]}{\term}{\fvar[2]}}[\mpar{\concrenv', \tripleset}]\\
    \mpar{
    \begin{rgathered}
      \concrfilter{F}{\myvec[\concrenv]{\svar{x}}{n}}{\myvecp{\cval}{m}}\\
      \concrenv' = \extsigvec{\concrenv}{\svar{y}}{\cval}[m]
    \end{rgathered}} &\implies
    \concrinterp{\filter{F}{\myvec{\svar{x}}{n}}{\myvecp{\svar{y}}{m}}}[\mpar{\concrenv', \tripleset}]\\
    \mpar{
    \begin{rgathered}
      \outfun\funu{i} = \concrenv_i\\
      \bvar \subseteq \dom{\concrenv_i}\\
      \concrenv' = \concrenv \extsymbol \envrestr{\concrenv_i}
    \end{rgathered}} &\implies
    \concrinterp[\outfun,\mpar{\concrenv,\tripleset}]{\bigoplus_n}[\mpar{\concrenv', \tripleset}][\bvar]
  \end{align*}
  \caption{Concrete Interpretation}\label{fig:concrete-interpretation}
\end{figure}

\paragraph{Running Example} We instantiate the base sort \(\identS\) with
strings and \(\litS\) with integers. We instantiate the flow sort \(\intT\) with
integers, \(\boolT\) with Booleans, \(\valT\) with the disjoint union \(\intT +
\boolT\), and \(\stateT\) with a partial function from strings to \(\valT\). The
concrete interpretation of the filters are the following partial functions:
\texttt{litInt} is the identity on integers, \(\texttt{intVal}\) and
\(\texttt{boolVal}\) inject their arguments in \(\intT + \boolT\),
\texttt{read}\(\funu{id,st}\): applies \(st\) to \(id\) (since \(st\) is a
partial function, it may not return a result), \texttt{isInt}\(\funu{v}\):
matches \(v\) in the disjoint union \(\intT + \boolT\), returns \(v\) if it is
in \(\intT\), \texttt{add}\(\funu{i_1,i_2}\): returns the integer addition of
\(i_1\) and \(i_2\), \texttt{eq}\(\funu{i_1,i_2}\): returns \texttt{true} if
\(i_1 = i_2\), \texttt{false} otherwise, \texttt{isBool}\(\funu{v}\): matches
\(v\) in the disjoint union \(\intT + \boolT\), returns \(v\) if it is in
\(\boolT\), \texttt{write}\(\funu{id,st,v}\): returns the partial function
mapping \(id\) to \(v\) and any other \(id'\) to \(st\funu{id'}\),
\texttt{id}\(\funu{st}\): returns \(st\), \texttt{isTrue}\(\funu{b}\): returns
\(()\) if \(b = \texttt{true}\), \texttt{isFalse}\(\funu{b}\): returns \(()\) if
\(b = \texttt{false}\). In the rest of the paper, we directly write \(\jsid{x}\)
for \(\var{\jsid{x}}\) and \(n\) for \(\const{n}\) in the examples.

\subsection{Consistency of WF and Concrete
  Interpretations}\label{subsec:type-concrete-consistent}

\begin{definition}
  We say a triple set \(\tripleset\) is \emph{well formed} if all its elements
  are well formed, i.e., if \(\mpar{\sigma,\term,\cval} \in \tripleset\), then
  \(t = c\myvecp{\term}{n}\) and there is a sort \(\sort\) such that
  \(\Sort[]{\term} = \sort\), \(\sigma : \typin{\sort}\), and \(\cval :
  \typout{\sort}\).
\end{definition}

We define \(\OKst{\mpar{\typenv,\dfvarset}}{\mpar{\concrenv,\tripleset}}\) as
follows: \(\tripleset\) is well-formed, \(\dom{\typenv} = \dom{\concrenv}\), and
for any \(x \in \dom{\typenv}\) we have \(\concrenv\funu{x} : \typenv\funu{x}\).
We define \(\OKout{\mpar{\typenv,\dfvarset}}{\concrenv}\) as follows: \(\dom{\typenv} =
\dom{\concrenv}\) and for any \(x \in \dom{\typenv}\) we have
\(\concrenv\funu{x} : \typenv\funu{x}\).

\begin{lemma}\label{lem:wf-concr-consistent}
    The well-formedness and concrete interpretations are universally consistent.

\end{lemma}

\subsection{Concrete Derivations}\label{subsec:concrete:derivations}

The concrete interpretation describes how skeletons can be interpreted from a set of
hooks. The \emph{immediate consequence} \(\cinterpf{}\) describes how skeletons can
be assembled. It starts from a set of well-formed triples (that is, of Hoare
triples) \(\tripleset\), and derives a new set of judgements using the concrete
interpretation. Intuitively, from the set of triples generated by derivations of
depth at most~\(n\), it builds the set of triples generated by derivations of
depth at most~\(n + 1\). It is defined as follows.
\begin{equation*}
  \cinterpf{\tripleset} = \mset[(\cstate,\term,\cval)]{
    \begin{gathered}
      \term = c\myvecp{\term}{n} \land
      \Sort[]{\term} = \sort\\
      \Rule{Name}{c\myvecp{\tvar}{n}}{\skel} \in \Rules\\
      \cstate : \typin{\sort}\\
      \concrenv = \extsigvec{\ivar \mapsto \sigma}{\tvar}{\term}[n]\\
      \concrinterp{\skel}[\concrenv']\\
      \concrenv'\funu{\ovar} = \cval
    \end{gathered}
  }
\end{equation*}

\begin{lemma}
  The functional \(\cinterpf{}\) is monotone.
\end{lemma}
\begin{proof}
  This is immediate by inspecting the interpretation of skeletal bodies, as the
  only one where \(\tripleset\) is used is for hooks, and a bigger
  \(\tripleset\) does not remove results.
\end{proof}

\begin{lemma}\label{lem:concr-wf}
    If \(\tripleset\) is a well-formed triple set, then \(\cinterpf{\tripleset}\)
  is a well-formed triple set.

\end{lemma}

We now show that the smallest fixpoint of~\(\cinterpf{}\) corresponds to the set of
triples generated by any finite derivation, or in other words, an inductive
definition of the concrete rules.

\begin{lemma}\label{lem:continuity}
    \(\cinterpf{}\) is continuous:
    for any increasing sequence of triple sets \(\mpar{T_i}\),
    we have \(\bigcup_i \cinterpf{T_i} = \cinterpf{\bigcup_i T_i}\).
\end{lemma}
\begin{proof}
    We prove this result by double inclusion.
    The inclusion \(\bigcup_i \cinterpf{T_i} \subseteq
    \cinterpf{\bigcup_i T_i}\) follows from the monotony of
    \(\cinterpf{}\). 
    To show that \(\cinterpf{\bigcup_i T_i} \subseteq \bigcup_i \cinterpf{T_i}\),
    we show that for all triple set \(T\) and \(\mpar{\sigma, t, o} \in \cinterpf{T}\),
    there exists a finite subset \(T'\) of \(T\) such that \(\mpar{\sigma, t, o} \in \cinterpf{T'}\).
    This result is immediate by induction over the structure of the skeleton \(\skel\).
    Then, for each \(\mpar{\sigma, t, o} \in \cinterpf{\bigcup_i T_i}\),
    there exists a finite subset \(T'\) of \(\bigcup_i T_i\) such that
    \(\mpar{\sigma, t, o} \in \cinterpf{T'}\).
    As \(T'\) is finite and \(\mpar{T_i}\) monotone,
    there exists \(n\) such that \(T' \subseteq T_i\).
    We conclude by monotonicity of \(\cinterpf{}\).
\end{proof}

\begin{definition}\label{def:csem}
  The concrete semantics~\(\csem\) is the smallest fixpoint of~\(\cinterpf{}\).
\end{definition}
\begin{lemma}\label{lem:csem}
    We have \({\csem} = \bigcup_n \cinterpf[n]{\emptyset}\).
\end{lemma}
\begin{proof}
  The set of triple sets ordered by inclusion is a CPO, and
  \(\cinterpf{}\) is continuous on this CPO. We conclude by Kleene fixpoint
  theorem.
\end{proof}

\begin{lemma}
  The concrete semantics \(\csem\) is well-formed.
\end{lemma}
\begin{proof}
  Let \(\mpar{\cstate,\term,\cval} \in {\csem}\). By Lemma~\ref{lem:csem}, there exists a finite number~\(n\) such that
  \(\mpar{\cstate,\term,\cval} \in \cinterpf[n]{\emptyset}\). We prove by
  induction on \(n\) that \(\mpar{\cstate,\term,\cval}\) has the expected
  properties. It is immediate for \(0\), and for \(n+1\) we simply apply
  lemma~\ref{lem:concr-wf}.
\end{proof}

\section{Abstract Interpretation}\label{SEC:ABSTRACT-INTERPRETATION}

This section describes how a set of skeletons defining a
programming language can be re-interpreted over an abstract domain of
properties to obtain an abstract interpretation of the language.

\subsection{Abstract Domains}\label{subsec:abstract:domains}

An abstract interpretation of a set of skeletons must define 
abstract domains for all the terms and flow sorts used in the skeleton bodies, ending
with abstract semantic states and abstract results.

Elements in the abstract domains represent sets of values in the
corresponding concrete domain (they are related through the concretion function
\(\concr\) introduced below). 
The abstract interpretation framework is designed to be parametric in the choice of abstract domains for base
values such as integers, Booleans, and program states. All
we require is that each abstract domain for sort \(\sort\) is a partial order \(\sqsubseteq\)
with a least element, denoted \(\bot_{\sort}\), representing
the empty set.
%
For example, the lattice of intervals can be used as an abstract
domains for integers, with $\bot_{\intT}$ being the empty
interval. Similarly, a state that maps program variables to
integer values can be abstracted as a mapping from variables to
intervals, or as a polyhedron that defines linear relations between
program variables. 

Skeletal variables can also range over terms. For each program or base sort \(\sort\),
we define an abstract domain by imposing a flat 
partial order on the set of terms of that sort (\emph{i.e.}, we relate a term to itself and no
other term) and by adding a $\bot_{\sort}$ element, smaller than all
terms of that sort. Abstract base terms include every concrete base term, they
may also include additional terms that denote sets of concrete base terms.
To ease notation, we sometimes omit the sort in \(\bot\) in an equality. In
this case, \(\aval = \bot\) should be read \(\aval = \bot_{\Sort[]{\aval}}\) 
and, \(\aval \neq \bot\) should be read \(\aval \neq \bot_{\Sort[]{\aval}}\).


\subsection{Abstract Interpretation of Skeletons}\label{subec:abstract:interpretation:rules}

In addition to the abstract domains, an abstract interpretation must specify its
input and output states, and how the empty skeleton body, hooks, filters, and
the merging of branches are interpreted. For each filter symbol \(F\) of
signature \(\myvecp{\sort}{n}\rightarrow\myvecp{\sortp}{m}\) we assume a
total function \(\abstrfilter{F}{}{}\) from the domain corresponding to
\(\myvecp{\sort}{n}\) to the domain corresponding to \(\myvecp{\sortp}{m}\). A
filter interpretation may return \(\bot\) to state it is not defined for that input.

\begin{figure}
  \centering
  \begin{align*}
 &\implies \abstrinterp{\emptylist}[\absflag][\abstrenv]\\
      \abstrenvp = \extenv{\abstrenv}{\fvar[2]}{\bot_{\typout{\Sort[]{\abstrenv\funu{t}}}}}
    &\implies \abstrinterp[\halt]{\deriv{\fvar[1]}{\term}{\fvar[2]}}[\halt][\abstrenvp, \abstripleset]\\
    \mpar{\begin{rgathered}
        \abstrenv\funu{\fvar[1]} \sqsubseteq \astate\\
        \abstrenv\funu{t} \sqsubseteq \aterm\\
        \mpar{\astate,\aterm,\bot} \in \abstripleset\\
      \abstrenvp = \extenv{\abstrenv}{\fvar[2]}{\bot_{\typout{\Sort[]{\abstrenv\funu{t}}}}}
    \end{rgathered}}
    &\implies \abstrinterp[\cont]{\deriv{\fvar[1]}{\term}{\fvar[2]}}[\halt][\abstrenvp, \abstripleset]\\
    \mpar{\begin{rgathered}
        \abstrenv\funu{\fvar[1]} \sqsubseteq \astate\\
        \abstrenv\funu{t} \sqsubseteq \aterm\\
        \mpar{\astate,\aterm,\aval} \in \abstripleset\\
        \aval \sqsubseteq \avalp\\
        \abstrenvp = \extenv{\abstrenv}{\fvar[2]}{\avalp}
    \end{rgathered}}
    &\implies \abstrinterp[\cont]{\deriv{\fvar[1]}{\term}{\fvar[2]}}[\cont][\abstrenvp, \abstripleset]\\
    \mpar{
    \begin{rgathered}
      \typsig{F} = \myvecp{\sortp}{n} \rightarrow \myvecp{\sort}{m}\\
      \abstrenvp = \extsigvec{\abstrenv}{\svar{y}}{\abot{\sort}}[m]
    \end{rgathered}}
    &\implies
    \abstrinterp[\halt]{\filter{F}{\myvec{\svar{x}}{n}}{\myvecp{\svar{y}}{m}}}[\halt][\abstrenvp, \abstripleset]
      \\
    \mpar{
    \begin{rgathered}
      \myvecp[\abstrenv]{\svar{x}}{n} \sqsubseteq \myvecp{\aval}{n}\\
      \abstrfilter{F}{\myvec{\aval}{n}}{\bot}\\
      \typsig{F} = \myvecp{\sortp}{n} \rightarrow \myvecp{\sort}{m}\\
      \abstrenvp = \extsigvec{\abstrenv}{\svar{y}}{\abot{\sort}}[m]
    \end{rgathered}}
    &\implies
    \abstrinterp[\cont]{\filter{F}{\myvec{\svar{x}}{n}}{\myvecp{\svar{y}}{m}}}[\halt][\abstrenvp, \abstripleset]
      \\
    \mpar{
    \begin{rgathered}
      \myvecp[\abstrenv]{\svar{x}}{n} \sqsubseteq \myvecp{\aval}{n}\\
      \abstrfilter{F}{\myvec{\aval}{n}}{} \sqsubseteq \myvecp{\avalp}{m}\\
      \abstrenvp = \extsigvec{\abstrenv}{\svar{y}}{\avalp}[m]
    \end{rgathered}}
    &\implies
    \abstrinterp[\cont]{\filter{F}{\myvec{\svar{x}}{n}}{\myvecp{\svar{y}}{m}}}[\cont][\abstrenvp, \abstripleset]
      \\
    \mpar{
    \begin{rgathered}
      n \geq 1\\
      \forall i \in [1..n]. \outfun\funu{i} = \mpar{\halt,\abstrenv[i]}\\
      \forall i \in [1..n]. \bvar \subseteq \dom{\abstrenv[i]}\\
      \forall i,j \in [1..n]. \envrestr{\abstrenv[i]} = \envrestr{\abstrenv[j]}\\
      \abstrenvp = \abstrenv \extsymbol \envrestr{\abstrenv[1]}
    \end{rgathered}}
    &\implies
      \abstrinterp[\absflag][\outfun,\mpar{\abstrenv,\abstripleset}]{\bigoplus_n}[\halt][\abstrenvp, \abstripleset][\bvar]
      \\
    \mpar{
    \begin{rgathered}
      \dom{\outfun} = [1..n]\\
      \mathcal{E} = \mset[\abstrenv_i]{\outfun\funu{i} = \mpar{\cont,\abstrenv[i]}} \neq \emptyset\\
      \forall \abstrenv[i] \in \mathcal{E}. \bvar \subseteq \dom{\abstrenv[i]}\\
      \abstrenv[i] \in \mathcal{E} \implies \abstrenvp = \abstrenv \extsymbol \envrestr{\abstrenv[i]}
    \end{rgathered}}
    &\implies
      \abstrinterp[\cont][\outfun,\mpar{\abstrenv,\abstripleset}]{\bigoplus_n}[\cont][\abstrenvp, \abstripleset][\bvar]
  \end{align*}
  \caption{Abstract Interpretation}\label{fig:abstract-interpretation}
\end{figure}

The input state of an abstract interpretation is a triple
\(\mpar{\absflag,\abstrenv,\abstripleset}\) comprising a \emph{flag}
\(\absflag\), an abstract environment $\abstrenv$ (mapping skeletal
variables to abstract terms and values), and a set of abstract semantic
triples $\abstripleset$ that gives semantics to hooks. A flag is either \(\halt\)
or \(\cont\) and it indicates whether it has been determined that the current
skeleton does not apply (\(\halt\)) or that it may still apply~(\(\cont\)). The
output state of an abstract interpretation is a flag and an abstract environment
where skeletal variables hold the result of the abstract interpretation.
Figure~\ref{fig:abstract-interpretation} defines the abstract semantics.

The abstract interpretation of an empty list of hypotheses just returns the flag
and environment from its input. There are three cases for the interpretation of
a hook $\deriv{\fvar[1]}{t}{\fvar[2]}$. If we have determined that the skeleton
does not apply, we set \(\fvar[2]\) to \(\bot\) of the correct sort. In the two
other cases, we need to have a triple \(\mpar{\astate,\aterm,\aval}\) from
\(\abstripleset\) such that \(\abstrenv\funu{\fvar[2]} \sqsubseteq \astate\) and
\(\abstrenv\funu{\term} \sqsubseteq \aterm\). This loss of precision gives some
flexibility for such a derivation. We then have two (non exclusive) cases: if
\(\aval = \bot\), then we know the skeleton does not apply, and set the flag
to \(\halt\) and \(\fvar[2]\) to the appropriate \(\bot\). For the last case, we
do not restrict what \(\aval\) is (it may still be \(\bot\)), and we bind in the
resulting environment \(\fvar[2]\) to some \(\avalp\) that may be less precise
than \(\aval\), again to gain flexibility.

The abstract interpretation of a filter
\(\filter{F}{\myvec{\svar{x}}{n}}{\myvecp{\svar{y}}{m}}\) also has three cases.
If we know the skeleton does not apply, we just bind the output variables
\(\myvecp{\svar{y}}{m}\) to the appropriate \(\bot\) depending on the signature
of \(F\). Otherwise, we apply the filter interpretation to an approximation of
the arguments as given by the environment. If the result is \(\bot\), we know
the skeleton does not apply and switch the flag to \(\halt\), as well as extend
the environment with \(\bot\) of the correct sort. Otherwise, we keep the flag
as \(\cont\) and extend the environment to an approximation of the result of the
filter.

For the merging operator, we interpret every possible branch and collect their
results in \(\outfun\). If all branching have the \(\halt\) flag (either because
the \(\halt\) flag was set before their interpretation, which would then be
propagated, or because they newly returned it), then the skeleton does not apply
and we set the flag accordingly, extending the environment with mappings from
the shared skeletal variables \(\bvar\) to \(\bot\) of the correct sort.
Otherwise, we collect all branches that have a \(\cont\) flag. They must all
return abstract environments that agree on the shared variables (which is why
the approximations in the filter and hook cases are useful, to ensure this is
possible), and we extend the current environment with this common environment.

The key difference between the abstract and concrete interpretations is
how the different results are merged in case of branching. The concrete
semantics picks one of them, whereas the abstract semantics requires all
branches that provided a result to agree. This is because the goal of the
abstract semantics is to infer abstract semantic triples that are valid
statements about all possible resulting states, i.e., about all possible
concrete choices in case of branching.

\subsection{Consistency of WF and Abstract Interpretations}\label{subsec:type-abstract-consistent}
We define \(\OKst{\mpar{\typenv,\dfvarset}}{\mpar{\absflag,\abstrenv,\abstripleset}}\) as
follows: \(\abstripleset\) is well formed, \(\dom{\typenv} = \dom{\abstrenv}\), and
for any \(x \in \dom{\typenv}\) we have \(\abstrenv\funu{x} : \typenv\funu{x}\).
We define \(\OKout{\mpar{\typenv,\dfvarset}}{\mpar{\absflag,\abstrenv}}\) as follows: \(\dom{\typenv} =
\dom{\abstrenv}\) and for any \(x \in \dom{\typenv}\) we have
\(\abstrenv\funu{x} : \typenv\funu{x}\).

\begin{lemma}\label{lem:wf-abstr-consistent}
    The well-formedness and abstract interpretations are universally consistent.

\end{lemma}

\subsection{Abstract Derivations}\label{subsec:abstract:derivations}

We define the abstract immediate consequence operator from well-formed triple
sets to triple sets in Figure~\ref{fig:abstract:immediate:consequence}.
As in the concrete case,
the immediate consequence describes how to assemble skeletons.

\begin{figure}
\begin{equation*}
  \ainterpf{\abstripleset} = \mset[(\astate,\aterm,\aval)]{
    \begin{gathered}
      \aterm = c\myvecp{\aterm}{n} \land
      \Sort[]{\aterm} = \sort\\
      \Rule{Name}{c\myvecp{\tvar}{n}}{\skel} \in \Rules\\
      \astate : \typin{\sort}\\
      \abstrenv = \extsigvec{\ivar \mapsto \astate}{\tvar}{\aterm}[n]\\
      \abstrinterp[\cont]{\skel}[\absflag][\abstrenvp]\\
      \abstrenvp\funu{\ovar} = \aval
    \end{gathered}}
\end{equation*}
    \caption{The abstract immediate consequence operator}
    \label{fig:abstract:immediate:consequence}
\end{figure}

\begin{lemma}
  The functional \(\ainterpf{}\) is monotonic.
\end{lemma}
\begin{proof}
  This is immediate by inspecting the interpretation of skeletal bodies, as the
  only one where \(\abstripleset\) is used is for hooks, and a bigger
  \(\abstripleset\) does not remove results.
\end{proof}

\begin{lemma}\label{lem:abstr-wf}
    If \(\abstripleset\) is a well-formed triple set, then
  \(\ainterpf{\abstripleset}\) is a well-formed triple set.

\end{lemma}

An abstract semantics \(\asem\) is a set of facts of the form
\(\mpar{\astate,\aterm,\aval}\) stating that from state \(\astate\) term
\(\aterm\) evaluates to \(\aval\). A \emph{correct} abstract semantics is one where
such triples correspond to triples in the concrete semantics (see
Section~\ref{subsec:concrete-abstract-consistent}). The more facts an abstract
semantics contains, the more useful it is, as it provides more information about
the behaviour of terms. Hence, we choose as abstract semantics the
one with most facts, i.e., the greatest fixpoint of
\(\ainterpf{}\). This choice provides a proof technique: since the greatest fixpoint is the
union of all sets such that \(\abstripleset \subseteq
\ainterpf{\abstripleset}\), to prove that a fact is correct, one can propose a
candidate set \(\abstripleset\) containing this fact, and then show that
\(\abstripleset \subseteq \ainterpf{\abstripleset}\). This amounts to proving
that the facts \(\abstripleset\) constitute an invariant of the semantics. If we were
to translate such invariants into a derivation, the resulting derivation may
be infinite.\footnote{See the Figure~10 of~\cite{schmidt1997abstract} for an
example of such representation.}

We could also define the abstract semantics \(\asem\) as the smallest
fixpoint of \(\ainterpf{}\). This would be sound but, having fewer
facts, we would then miss valuable abstract results. More precisely,
if a triple \(\mpar{\astate,\term,\aval}\) belongs to the smallest
fixpoint of \(\ainterpf{}\), then (as abstract triple sets form a CPO
ordered by inclusion and \(\ainterpf{}\) is continuous on this CPO),
there exists a finite number~\(n\) such that
\(\mpar{\astate,\term,\aval} \in \ainterpf[n]{\emptyset}\). In other
words, there exists a finite abstract derivation yielding the triple
\(\mpar{\astate,\term,\aval}\). This implies that for all concrete
state \(\cstate{} \in \gamma\funu{\astate}\), the program \(\term\)
terminates. We would thus have lost all facts for which the abstract
semantics cannot prove termination.  Defining the abstract semantics
as the greatest fixpoint of \(\ainterpf{}\) solves this issue.

\begin{definition}\label{def:asem}
  The abstract semantics~\(\asem\) is the largest fixpoint of~\(\ainterpf{}\)
  as a function from well-formed triple sets to well-formed triple sets.
  This restriction is well-defined by Lemma~\ref{lem:abstr-wf}.
\end{definition}
\begin{lemma}\label{lem:asem-wf}
  \(\asem\) is well formed.
\end{lemma}
\begin{proof}
  As \(\asem\) is the largest fixpoint of \(\ainterpf{}\), it is the union of all
  well-formed triple sets \(\abstripleset\) such that \(\abstripleset \subseteq
  \ainterpf{\abstripleset}\). Let \(\mpar{\astate,\aterm,\aval} \in \asem\),
  there is \(\abstripleset \subseteq \ainterpf{\abstripleset}\) where
  \(\mpar{\astate,\aterm,\aval} \in \abstripleset\) and \(\abstripleset\) is
    well formed.
    Hence \(\mpar{\astate,\aterm,\aval}\) has the requested properties.
\end{proof}

\subsection{Consistency of Concrete and Abstract
  Interpretations}\label{subsec:concrete-abstract-consistent}

We assume a \emph{concretion function}
\(\concr\) for the abstract domain,
from abstract terms to sets of concrete terms,
and from abstract values to sets of concrete values. 
We impose several constraints on \(\concr\).
First, \(\concr\) must be compatible with \(\sqsubseteq\): if \(\term \in
\concr\funu{\aterm}\) and \(\aterm \sqsubseteq \atermp\), then \(\term \in
\concr\funu{\atermp}\), and if \(\cval \in \concr\funu{\aval}\) and \(\aval
\sqsubseteq \avalp\), then \(\cval \in \concr\funu{\avalp}\).
Second, for any abstract term \(\aterm\) of sort \(s\), the set
\(\concr\funu{\aterm}\) must only contain terms of sort \(s\).
In addition,
\(\concr\funu{c\myvecp{\aterm}{n}} = \mset[c\myvecp{\term}{n}]{\term[i] \in
  \concr\funu{\aterm[i]}}\). Conversely, for any concrete term \(\term\), we
have \(\concr\funu{\term} = \mset{\term}\), as abstract base terms are
extensions of concrete base terms.

\begin{lemma}\label{lem:concr-terms}
  Let \(\concrenv\) be a mapping from term variables to concrete terms, and
  \(\abstrenv\) be a mapping from term variables to abstract terms. If
  \(\Tvar{t} \subseteq \dom{\concrenv}\), \(\Tvar{t} \subseteq
  \dom{\abstrenv}\), and \(\forall \tvar \in \Tvar{t}, \concrenv\funu{\tvar} \in
  \concr\funu{\abstrenv\funu{\tvar}}\), then \(\concrenv\funu{t} \in
  \concr\funu{\abstrenv\funu{t}}\).
\end{lemma}
\begin{proof}
  By induction on the structure of \(\term\). If it is a base term, then the
  result holds by hypothesis on base terms, if it is a term variable, then the
  result is immediate, and otherwise we prove the property by induction on the
  subterms.
\end{proof}

Regarding values, we have similar restrictions: for any abstract value \(\aval\)
of sort \(\sort\), the concrete values in \(\concr\funu{\aval}\) all have sort
\(\sort\). We also require the abstract interpretation of filters to be
consistent with the concrete one: if
\(\concrfilter{F}{\myvec{\cval}{n}}{\myvecp{\cvalp}{m}}\) and \(\forall i \in
[1..n].\cval[i] \in \concr\funu{\aval[i]}\), then 
\(\abstrfilter{F}{\myvec{\aval}{n}}{\myvecp{\avalp}{m}}\) and \(\forall i \in
[1..n].\cvalp[i] \in \concr\funu{\avalp[i]}\). In particular, if the concrete
filter relates its input to an output, the abstract filter cannot return
\(\bot\).

\begin{definition}
  Let \(\tripleset\) a concrete triple set and \(\abstripleset\) an abstract
  triple set. We say they are \emph{consistent} if for any \(\mpar{\cstate,
    \term, \cval} \in \tripleset\) and \(\mpar{\astate, \aterm, \aval} \in
  \abstripleset\), if \(\cstate \in \concr\funu{\astate}\) and \(\term \in
  \concr\funu{\aterm}\), then \(\cval \in \concr\funu{\aval}\).
\end{definition}

We define
\(\OKst{\mpar{\concrenv,\tripleset}}{\mpar{\absflag,\abstrenv,\abstripleset}}\)
as follows: \(\absflag = \cont\), \(\dom{\concrenv} = \dom{\abstrenv}\), for any
\(\svar{x} \in \dom{\concrenv}\), we have \(\concrenv\funu{\svar{x}} \in
\concr\funu{\abstrenv\funu{\svar{x}}}\), and \(\tripleset\) and
\(\abstripleset\) are well formed and consistent.
We define \(\OKout{\concrenv}{\mpar{\absflag,\abstrenv}}\) as follows: \(\absflag
= \cont\), \(\dom{\concrenv} = \dom{\abstrenv}\), and for any \(\svar{x} \in
\dom{\concrenv}\), we have \(\concrenv\funu{\svar{x}} \in
\concr\funu{\abstrenv\funu{\svar{x}}}\).

\begin{lemma}\label{lem:concr-abstr-consistent}
    The concrete and abstract interpretations are universally consistent.

\end{lemma}

\begin{lemma}\label{lem:abstr-concr-wf-consistent}
    Let \(\tripleset\) and \(\abstripleset\) well formed and consistent triple
  sets,  then
  \(\cinterpf{\tripleset}\) and \(\ainterpf{\abstripleset}\) are well formed and
  consistent triple sets.

\end{lemma}

We finally show that the abstract semantics is correct relative to the
concrete semantics. In a nutshell, for any triple in the concrete
semantics \(\csem\) (the smallest fixpoint of \(\cinterpf{}\)) and any
triple in the abstract semantics \(\asem\) (the largest fixpoint of
\(\ainterpf{}\)), if the input states and terms are related, then the
output values are related. Formally, we have the following.

\begin{definition}
  An abstract triple set \(\abstripleset\)
  is \emph{correct} if
  it is well-formed and consistent with \(\csem\).
\end{definition}

\begin{theorem}
  \(\asem\) is correct.
\end{theorem}
\begin{proof}
  We prove by induction on \(k\) that \(\cinterpf[k]{\emptyset}\) and \(\asem\)
  are well formed and consistent. The check that \(\asem\) is well formed is
  simply Lemma~\ref{lem:asem-wf}.

  The result is immediate for \(k = 0\) since \(\emptyset\) is well formed, and
  there is nothing else to check.

  Let \(k = n+1\), by induction we have \(\cinterpf[n]{\emptyset}\) and \(\asem\)
  are well formed and consistent. By Lemma~\ref{lem:concr-wf}
  we have \(\cinterpf[n+1]{\emptyset}\) and \(\ainterpf{\asem} = \asem\) are well
  formed and consistent, as required.

  To conclude, we apply Lemma~\ref{lem:csem}.
\end{proof}

Note that in the previous theorem we only use the fact that \(\asem\) is a
fixpoint: it does not have to be the greatest fixpoint.

\subsection{Example: Interval analysis of \While}\label{ex:abstraction:example}

\begin{figure}
\begin{align*}
    \absint{\mathtt{litInt}}
  &  = \lambda (n). \itvl{n}{n}
    &
    \absint{\mathtt{id}}
    & = \lambda \astate.\ \astate
\\
  \absint{\mathtt{intVal}} & = \lambda i. \mpar{i,\bot_{\boolT}}
&
  \absint{\mathtt{boolVal}} & = \lambda b. \mpar{\bot_{\intT},b}\\
  \absint{\mathtt{isInt}}
  & = \lambda (i,b).\ i
&
  \absint{\mathtt{isBool}}
  & = \lambda (i,b).\ b
\\
  \absint{\mathtt{add}}
  & =
  \lambda (\itvl{l_1}{u_1},\itvl{l_2}{u_2}). \itvl{l_1 + l_2}{u_1 + u_2}
   \\
  \absint{\mathtt{eq}}
  &  = \lambda (i_1,i_2).\
  \begin{cases}
    \abs\btrue & \text{if \(i_1 = \itvl{n}{n} = i_2\)} \\
    \abs\bfalse & \text{if \(i_1 \cap i_2 = \emptyset\)} \\
    \top_{\boolT} & \text{otherwise}
  \end{cases}
&
  \absint{\mathtt{neg}}
  & = \lambda b.\
    \begin{cases}
        \abs\bfalse & \text{if \(b = \abs\btrue\)} \\
        \abs\btrue & \text{if \(b = \abs\bfalse\)} \\
        b & \text{otherwise}
    \end{cases}
\\
  \absint{\mathtt{read}}
    & = \lambda \mpar{\astate, \jsid{x}}.\ \eread{\astate}{\jsid{x}}
&
    \absint{\mathtt{write}}
    & = \lambda (\jsid{x},\astate,\abs{v}).\ \ewrite{\jsid{x}}{\astate}{\abs{v}}
\\
    \absint{\mathtt{isTrue}}
    & = \lambda b.\
    \begin{cases}
        \bot & \text{if \(b \in \mset{\bot_{\boolT}, \abs\bfalse}\)} \\
        \unit & \text{otherwise}
    \end{cases}
&
    \absint{\mathtt{isFalse}}
    & = \lambda b.\
    \begin{cases}
        \bot & \text{if \(b \in \mset{\bot_{\boolT}, \abs\btrue}\)} \\
        \unit & \text{otherwise}
    \end{cases}
\end{align*}
    \caption{Abstract interpretation of filters}
    \label{fig:example:abstract:interpretation:filters}
\end{figure}

To give a concrete example of an abstract interpretation we design a
value analysis of the language \While{} in the style of Schmidt's Abstract
Interpretation of Natural Semantics~\cite{Schmidt:95:Natural}. We
have the following flow sorts in the semantic definition
(cf. Figure~\ref{fig:example:typing:filters}):
\(\intT\), \(\boolT\), \(\valT\), and \(\stateT\).

We describe an analysis in which integers are
approximated by intervals, ordered by inclusion.
Writing $[n,m]$ for the interval of integers between $n$ and $m$ (with
the convention that $[n,m] = \emptyset$ if $m < n$), we can define
the abstract domains for each of the flow sort as follows:
\begin{align*}
  \abs{\intT} & = \mpar{\itvl{n}{m} : n \in \mathbf{Z} \cup \mset{-\infty} \land m \in \mathbf{Z} \cup \mset{+\infty}} & 
  \abs{\valT} & = \abs{\intT} \times \abs{\boolT} \\
  \abs{\boolT} & = \{\bot_{\boolT},\abs\ptrue,\abs\pfalse,\top_{\boolT} \} &
  \abs{\stateT} & = \abs{\identS} \rightarrow \abs{\valT}
\end{align*}
Abstract base terms are concrete base terms. We abstract identifiers by themselves, $\abs{\identS} =
\identS$, with only the trivial (reflexive) ordering.
The abstract domain of Booleans is (isomorphic to) the set
of subsets of Booleans, ordered by inclusion.  The abstract domain of
values is the defined as the Cartesian product, ordered
component-wise, of the abstract domain of integers and Booleans, where
each component gives an approximation of the concrete value,
\emph{provided} that the value is of the corresponding sort.  Stores
are mappings from identifiers to values, ordered
pointwise. Undefined identifiers are mapped to the
undefined value $\bot_{\abs{\valT}}$.
The concretisation function  $\gamma$ from abstract domains to
concrete domains formalises the relation between concrete and abstract
values.
\begin{align*}
  \gamma\funu{[n,m]} &= \mset{i \st n \leq i \leq m} &
  \gamma\funu{i,b} &= \gamma\funu{i} \cup \gamma\funu{b}&
  \gamma\funu{\top_{\boolT}} & = \mset{\btrue,\bfalse}\\
  \gamma\funu{\bot_{\boolT}} &= \emptyset&
  \gamma\funu{\abs\ptrue} &= \mset{\btrue}&
  \gamma\funu{\abs\pfalse} &= \mset{\bfalse}\\&&
  \gamma\funu{\astate} &= \mset{\concrenv \st
\forall \jsid{x}.\concrenv\funu{\jsid{x}} \in
\gamma\funu{\astate\funu{\jsid{x}}}}
\end{align*}

The abstraction of the basic filters used in the definition of
\While{} is given in
Figure~\ref{fig:example:abstract:interpretation:filters}. Notice that
the abstract interpretation of the filters \texttt{isInt} and
\texttt{isBool} return an abstract integer and an abstract Boolean,
respectively, instead of a Boolean stating whether their argument can
be an integer and a Boolean. This is correct because an abstract value
that is only an integer has the shape \(\mpar{i,\bot_{\boolT}}\), and
applying \texttt{isBool} to it returns \(\bot_{\boolT}\), indicating
it contains no Boolean.

\begin{lemma}
  The abstract filters are consistent with the concrete filters.
\end{lemma}

\begin{lemma}
  The abstract semantics of \While{} is correct.
\end{lemma}

\section{Deriving Proof Techniques from an Abstract Semantics}\label{SEC:PROOF-TECHNIQUES}

This section presents several proof techniques derived from an abstract
semantics and instantiated in our \While{} language.

\subsection{Abstract rules for analysing \While}\label{subsec:example:abstract:rules:as:rules}

Given the instantiation of the filters used in the abstract semantic
of \While{}, we can now derive an abstract interpretation of
\While{} programs. The result of an abstract interpretation of a program
is a set of abstract triples that correctly describes the program
behaviour. We shall present the analysis through a set of
syntax-directed inference
rules for inferring such triples. For a given term
$c\myvecp{\term}{n}$, we take the corresponding skeleton
\(\Rule{Name}{c\myvecp{\tvar}{n}}{\skel}\) in the semantics and
apply the general abstract interpretation to the skeleton body
$\skel$. This results in a series of conditions for a triple to be
valid that will form the hypotheses of the inference rules.
%

\paragraph{Rule for addition}
As a first example,
we derive a rule for analysing arithmetic expressions such as
\(\term[1]{} + \term[2]{}\).
A triple $\mpar{\astate,\term[1]+\term[2],\aval}$ is valid if it belongs to a
fixpoint $\abstripleset$ of $\ainterpf{}$.
Unfolding definitions, 
\begin{align*}
& \mpar{\astate,\term[1]+\term[2],\aval} \in \abstripleset = \ainterpf{\abstripleset} \\
\Leftrightarrow\quad
    &  \abstrinterp[\cont][\abstrenv[1],\abstripleset]{
        \begin{lgathered}
        \deriv{\ivar{}}{\tvar[1]{}}{\fvar[1]{}};
        \filter{\texttt{isInt}}{\fvar[1]{}}{\fvar[1']};
        \deriv{\ivar{}}{\tvar[2]{}}{\fvar[2]{}};\\
        \filter{\texttt{isInt}}{\fvar[2]{}}{\fvar[2']{}};
        \filter{\texttt{add}}{\fvar[1']{}, \fvar[2']{}}{\fvar[3]{}};
        \filter{\texttt{intVal}}{\fvar[3]}{\ovar}
        \end{lgathered}
    }[\absflag][\abstrenv[o]] \\
&\quad \land \quad
    \abstrenv[1] = \ivar \mapsto \astate
                \extsymbol \tvar[1] \mapsto \term[1]
                \extsymbol \tvar[2] \mapsto \term[2]
\quad \land \quad
    \aval = \abstrenv[o]\funu{\ovar}
\end{align*}

For simplicity, we here choose to ignore weakenings and the
non-\(\cont\)-case for the flag \(\absflag\).  In other words, we are
ignoring the possibility of short-cutting the abstract interpretation
of the rule if a \(\bot\) is found during the abstract execution.
\begin{align*}
& \mpar{\astate,\term[1]+\term[2],\aval} \in \ainterpf{\abstripleset} \\
\Leftarrow\quad&
                 \mpar{\abstrenv[1]\funu{\ivar}, \abstrenv[1]\funu{\tvar[1]}, \abs{v}_1} \in \abstripleset \\
&\quad \land \quad
\abstrinterp[\cont][\abstrenv[2],\abstripleset]{
        \begin{lgathered}
        \filter{\texttt{isInt}}{\fvar[1]{}}{\fvar[1']};
        \deriv{\ivar{}}{\tvar[2]{}}{\fvar[2]{}};
        \filter{\texttt{isInt}}{\fvar[2]{}}{\fvar[2']{}};\\
        \filter{\texttt{add}}{\fvar[1']{}, \fvar[2']{}}{\fvar[3]{}};
        \filter{\texttt{intVal}}{\fvar[3]}{\ovar}
        \end{lgathered}
        }[\cont][\abstrenv[o]] \\
&\quad \land \quad
    \abstrenv[1] = \ivar \mapsto \astate
                \extsymbol \tvar[1] \mapsto \term[1]
                \extsymbol \tvar[2] \mapsto \term[2] 
\quad \land \quad
    \abstrenv[2] = \abstrenv[1]
                \extsymbol \fvar[1]{} \mapsto \abs{v}_{1}
\quad \land \quad
    \aval = \abstrenv[o]\funu{\ovar}
\end{align*}

Interpreting the filter \texttt{isInt} makes us consider the integer projection
of the abstract value \(\abs{v}_1\).
We can thus rewrite the implication as follows.
\begin{align*}
& \mpar{\astate,\term[1]+\term[2],\aval} \in \ainterpf{\abstripleset} \\
    \Leftarrow\quad&
                 \mpar{\astate, \term[1], \abs{v}_1} \in \abstripleset
    \quad \land \quad \abs{v}_{1} = \mpar{\abs{i}_{1}, \abs{b}_{1}} \\
 &\quad \land \quad
 \abstrinterp[\cont][\abstrenv[2],\abstripleset]{
   \begin{lgathered}
     \deriv{\ivar{}}{\tvar[2]{}}{\fvar[2]{}};
     \filter{\texttt{isInt}}{\fvar[2]{}}{\fvar[2']{}};\\
     \filter{\texttt{add}}{\fvar[1']{}, \fvar[2']{}}{\fvar[3]{}};
     \filter{\texttt{intVal}}{\fvar[3]}{\ovar}
   \end{lgathered}
         }[\cont][\abstrenv[o]] \\
 &\quad \land \quad
     \abstrenv[2] = \ivar \mapsto \astate
                 \extsymbol \tvar[1] \mapsto \term[1]
                 \extsymbol \tvar[2] \mapsto \term[2]
                 \extsymbol \fvar[1]{} \mapsto \abs{v}_{1}
                 \extsymbol \fvar[1']{} \mapsto \abs{i}_{1}
   \quad \land \quad
       \aval = \abstrenv[o]\funu{\ovar}
\end{align*}

We can continue unfolding the abstract interpretation of the rule.
We eventually reach the following implication:
\begin{align*}
& \mpar{\astate,\term[1]+\term[2],\aval} \in \ainterpf{\abstripleset} \\
\Leftarrow\quad&
 \mpar{\astate, \term[1], \abs{v}_1} \in \abstripleset ~ \land ~
    \mpar{\astate, \term[2], \abs{v}_2} \in \abstripleset
    ~ \land ~ \abs{v}_{1} = \mpar{\abs{i}_{1}, \abs{b}_{1}}
    ~ \land ~ \abs{v}_{2} = \mpar{\abs{i}_{2}, \abs{b}_{2}}
 ~ \land ~
 \aval = \abstrenv[o]\funu{\ovar}
\\
&\quad \land \quad
    \abstrinterp[\cont][\abstrenv[4],\abstripleset]{
     \filter{\texttt{add}}{\fvar[1']{}, \fvar[2']{}}{\fvar[3]{}};
     \filter{\texttt{intVal}}{\fvar[3]}{\ovar}
    }[\cont][\abstrenv[o]]
\\
&\quad \land \quad
    \abstrenv[4] = \ivar \mapsto \astate
                \extsymbol \tvar[1] \mapsto \term[1]
                \extsymbol \tvar[2] \mapsto \term[2]
                \extsymbol \fvar[1]{} \mapsto \abs{v}_{1}
                \extsymbol \fvar[1']{} \mapsto \abs{i}_{1}
                \extsymbol \fvar[2]{} \mapsto \abs{v}_{2}
                \extsymbol \fvar[2']{} \mapsto \abs{i}_{2}
\\
\Leftarrow\quad
&
    \mpar{\astate , \term[1]{}, \mpar{\abs{i}_1, \abs{b}_1}} \in \abstripleset
 ~ \land ~
    \mpar{\astate , \term[2]{}, \mpar{\abs{i}_2, \abs{b}_2}} \in \abstripleset
  ~ \land ~ \abstrfilter{add}{\abs{i}_{1}, \abs{i}_{2}}{\abs{i}} ~ \land
                  ~ \abs{v} = \mpar{\abs{i},\bot_{\boolT}}
\end{align*}

By writing $\astate \vdash t : \abs{v}$ for $\mpar{\astate, t , \abs{v}} \in \abstripleset$
we get the familiar rule below.
\begin{mathpar}
        \inferrule{
         \astate \vdash \term[1]{} : \mpar{\abs{i}_1, \abs{b}_1} \\
         \astate \vdash \term[2]{} : \mpar{\abs{i}_2, \abs{b}_2} \\
         \abstrfilter{add}{\abs{i}_{1},\abs{i}_{2}}{\abs{i}}}
         {\astate \vdash \term[1]{} + \term[2]{} : \mpar{\abs{i},\bot_{\boolT}}}
\end{mathpar}

\paragraph{Rule for conditionals}
In the case of the addition, the structure of the skeleton was linear.
We have seen that we ignored some branches
(the ones triggering \(\bot\)), but these were not very important.
We now show the example of conditionals,
where branches are more visible.
A triple $\mpar{\astate,\sif{\term[1]}{\term[2]}{\term[3]},\astate[o]}$
is valid if it belongs to a fixpoint $\abstripleset$ of $\ainterpf{}$.
Unfolding definitions,
and passing through the linear part of the skeleton, we get:
\begin{align*}
    & \mpar{\astate,\sif{\term[1]}{\term[2]}{\term[3]},\astate[o]} \in
  \ainterpf{\abstripleset} \\
\Leftrightarrow \quad& \\
\mathrlap{\hspace{-0.6cm}\abstrinterp[\cont][\abstrenv[1],\abstripleset]{
    \begin{lgathered}
        \deriv{\ivar{}}{\tvar[1]}{\fvar[1]};
        \filter{\texttt{isBool}}{\fvar[1]}{\fvar[1']};
        \branchesV[\mset{\ovar{}}]{\begin{aligned}
          &\filter{\texttt{isTrue}}{\fvar[1']}{};
          \deriv{\ivar{}}{\tvar[2]}{\ovar{}}\\
          &\filter{\texttt{isFalse}}{\fvar[1']}{};
          \deriv{\ivar{}}{\tvar[3]}{\ovar{}}
        \end{aligned}}
    \end{lgathered}
            }[\absflag][\abstrenv[o]]} & \hspace{\textwidth} \\
&\quad \land \quad
    \abstrenv[1] = \ivar \mapsto \astate
                \extsymbol \tvar[1] \mapsto \term[1]
                \extsymbol \tvar[2] \mapsto \term[2]
                \extsymbol \tvar[3] \mapsto \term[3]
\quad \land \quad
    \astate[o] = \abstrenv[o]\funu{\ovar} \\
\Leftarrow \quad &
    \mpar{\astate, \term[1], \abs{v}_1} \in \abstripleset
    \quad \land \quad \abs{v}_{1} = \mpar{\abs{i}_{1}, \abs{b}_{1}} \quad \land \\
&  \abstrinterp[\cont][\abstrenv[2],\abstripleset]{
    \begin{lgathered}
        \branchesV[\mset{\ovar{}}]{\begin{aligned}
          &\filter{\texttt{isTrue}}{\fvar[1']}{};
          \deriv{\ivar{}}{\tvar[2]}{\ovar{}}\\
          &\filter{\texttt{isFalse}}{\fvar[1']}{};
          \deriv{\ivar{}}{\tvar[3]}{\ovar{}}
        \end{aligned}}
    \end{lgathered}
    }[\absflag][\abstrenv[o]] \\
&\quad \land \quad
    \abstrenv[1] = \ivar \mapsto \astate
                \extsymbol \tvar[1] \mapsto \term[1]
                \extsymbol \tvar[2] \mapsto \term[2]
                \extsymbol \tvar[3] \mapsto \term[3]
                \extsymbol \fvar[1] \mapsto \abs{v}_{1}
                \extsymbol \fvar[1'] \mapsto \abs{b}_{1}
\end{align*}

From this stage, we continue the analysis
in each of the two subbranches to build a map \(\outfun\)
representing the outputs of both branches.
We consider two cases,
depending on the value of \(\abs{b}_{1}\).

First, if \(\abs{b}_{1}\) is \(\top_{\boolT}\).
We then have both \texttt{isTrue} and \texttt{isFalse}
holding on \(\abs{b}_{1}\).
By unfolding definitions and using weakening for the results of the two hooks, 
we get the following implication:
\begin{align*}
    & \mpar{\astate,\sif{\term[1]}{\term[2]}{\term[3]},\astate[o]} \in \ainterpf{\abstripleset} \\
\Leftarrow\quad&
    \mpar{\astate, \term[1], \abs{v}_1} \in \abstripleset
    \quad \land \quad \mpar{\astate, \term[2], \astate[2]} \in \abstripleset
    \quad \land \quad \mpar{\astate, \term[3], \astate[3]} \in \abstripleset \\&
    \quad \land \quad \abs{v}_{1} = \mpar{\abs{i}_{1}, \top_{\boolT}}
    \quad \land \quad \astate[2] \sqsubseteq \astate[o]
    \quad \land \quad \astate[3] \sqsubseteq \astate[o]
\end{align*}
Using the same notations as above
we can simplify this rule as below.
\begin{mathpar}
    \inferrule{
         \astate \vdash \term[1]{} : \mpar{\abs{i}_1, \top_{\boolT}} \\
         \astate \vdash \term[2]{} : \astate[2] \\
         \astate[2] \sqsubseteq \astate[o]\\
         \astate \vdash \term[3]{} : \astate[3]\\
         \astate[3] \sqsubseteq \astate[o]}
     {\astate \vdash \sif{\term[1]}{\term[2]}{\term[3]} : \astate[o]}
\end{mathpar}

This rule is imprecise
(we assume that we get \(\top_{\boolT}\) when evaluating the conditional's expression),
but shows how our equivalent of concrete rules
are merged in the abstract interpretation.
We now consider a more precise version of the rule,
for the case when the conditional expression evaluates to \(\abs{\btrue}\).
The other cases \(\abs\bfalse\) and \(\bot_{\boolT}\) are similar.
In this case, the \texttt{isTrue} filter holds,
but not \texttt{isFalse}:
we can derive the judgement below when \(\eread{\abstrenv}{\fvar[1']} = \abs{\btrue}\).
\[
    \abstrinterp[\cont][\abstrenv,\abstripleset]{
       \filter{\texttt{isFalse}}{\fvar[1']}{};
       \deriv{\ivar{}}{\tvar[3]}{\ovar{}}
    }[\halt][\abstrenv[o]]
\]
Following the rules for abstract interpretation (see Figure~\ref{fig:abstract-interpretation}),
this removes the second branch from the \(\mathcal{E}\) set,
only leaving constraints from the first branch.
We thus get the following implication, where we no longer need the weakening for
the result of the hook.
\begin{align*}
    & \mpar{\astate,\sif{\term[1]}{\term[2]}{\term[3]},\astate[o]} \in \ainterpf{\abstripleset} \\
\Leftarrow\quad&
    \mpar{\astate, \term[1], \abs{v}_1} \in \abstripleset
    \quad \land \quad \mpar{\astate, \term[2], \astate[o]} \in \abstripleset
    \quad \land \quad \abs{v}_{1} = \mpar{\abs{i}_{1}, \abs{\btrue}}
\end{align*}
We can rewrite this implication as above into the rule
\begin{mathpar}
    \inferrule{
        \astate \vdash \term[1]{} : \mpar{\abs{i}_1, \abs{\btrue}} \\
         \astate \vdash \term[2]{} : \astate[2]}
     {\astate \vdash \sif{\term[1]}{\term[2]}{\term[3]} :
       \astate[2]}.
\end{mathpar}

\paragraph{Rule for loops}
The skeleton for loops is close to the one for conditionals.
We can similarly derive abstract rules such as the ones below.
\begin{mathpar}
    \inferrule{
         \astate \vdash \term[1]{} : \mpar{\abs{i}_1, \top_{\boolT}} \\
         \astate \vdash \term[2]{} : \astate[2] \\
         \astate[2] \vdash \while{\term[1]}{\term[2]} : \astate[3] \\
         \astate[3] \sqsubseteq \astate}
     {\astate \vdash \while{\term[1]}{\term[2]} : \astate}
     \and
    \inferrule{
        \astate \vdash \term[1]{} : \mpar{\abs{i}_1, \abs{\btrue}} \\
         \astate \vdash \term[2]{} : \astate[2] \\
         \astate[2] \vdash \while{\term[1]}{\term[2]} : \astate[3]}
     {\astate \vdash \while{\term[1]}{\term[2]} : \astate[3]}
     \and
    \inferrule{
        \astate \vdash \term[1]{} : \mpar{\abs{i}_1, \abs{\bfalse}}}
     {\astate \vdash \while{\term[1]}{\term[2]} : \astate}
\end{mathpar}

We can also use the fact that any fixpoint of \(\ainterpf{}\)
is considered valid.
The following implication
(which we can prove in a way similar to above) is
valid for any well-formed set \(\abstripleset{}\).
\begin{align*}
    & \mpar{\astate,\while{\term[1]}{\term[2]},\astate[o]} \in \ainterpf{\abstripleset} \\
\Leftarrow\quad&
    \mpar{\astate, \term[1], \abs{v}_1} \in \abstripleset
    \quad \land \quad \abs{v}_{1} \sqsubseteq \mpar{\abs{i}_{1}, \top_{\boolT}} \\&
    \quad \land \quad \mpar{\astate, \term[2], \astate[2]} \in \abstripleset
    \quad \land \quad \mpar{\astate[2],\while{\term[1]}{\term[2]},\astate[3]} \in \abstripleset
    \quad \land \quad \astate[3] \sqsubseteq \astate[o]
    \quad \land \quad \astate \sqsubseteq \astate[o]
\end{align*}
In particular,
as the condition
\(\abs{v}_{1} \sqsubseteq \mpar{\abs{i}_{1}, \top_{\boolT}}\)
is vacuously true, we can weaken this implication as follows
(forcing all intermediate states to be the same).
\[
    \mpar{\astate,\while{\term[1]}{\term[2]},\astate} \in \ainterpf{\abstripleset}
\quad\Leftarrow\quad
    \mpar{\astate, \term[1], \abs{v}_1} \in \abstripleset
    \land \mpar{\astate, \term[2], \astate} \in \abstripleset
    \land \mpar{\astate,\while{\term[1]}{\term[2]},\astate} \in \abstripleset
\]
This implication means that given any \(\abstripleset[0]\),
such that \(\abstripleset[0] \subseteq \ainterpf{\abstripleset[0]}\),
that associates \(\term[1]\) in the state \(\astate\)
with a result (that is that there exists  \(\aval\) such that
\(\mpar{\astate, \term[1], \aval} \in \abstripleset[0]\)),
and such that \(\mpar{\astate, \term[2], \astate} \in \abstripleset[0]\),
we can extend \(\abstripleset[0]\) into
\(\abstripleset[1] = \abstripleset[0] \cup \mset{\mpar{\astate,\while{\term[1]}{\term[2]},\astate}}\).
By monotonicity of \(\ainterpf{}\), we get \(\abstripleset[0] \subseteq \ainterpf{\abstripleset[1]}\),
and by the above implication,
we get \(\mpar{\astate,\while{\term[1]}{\term[2]},\astate} \in \ainterpf{\abstripleset[1]}\).
Hence, \(\abstripleset[1] \subseteq \ainterpf{\abstripleset[1]}\), and every
triple in \(\abstripleset[1]\) is correct in relation to \(\csem\).
In other words, the following familiar rule is admissible.
\begin{mathpar}
    \inferrule{
        \astate \vdash \term[1]{} : \aval \\
         \astate \vdash \term[2]{} : \astate}
     {\astate \vdash \while{\term[1]}{\term[2]} : \astate}
\end{mathpar}

\subsection{State Splitting}
As another example of the use of the abstract interpretation, we show
how to extend the abstract semantics to obtain more precise
results. Our motivating example is \(t\): $\while{\neg\mpar{\jsid{x}
    = 0}}{\asnescape{\jsid{x}}{\jsid{x} - 1}}$ for which we want to show that the
triple \(\mpar{\jsid{x} \mapsto \itvl{0}{\infty}, t, \jsid{x} \mapsto 0}\) is
correct (we simplify notation and write \(n\) for
\(\mpar{\itvl{n}{n},\abot{\boolT}}\), and \(\itvl{n}{m}\) for
\(\mpar{\itvl{n}{m},\abot{\boolT}}\)). Proving this is not possible as
such. To see this, observe that in the
rule for \rulen{While}, the same state is used to run the
expression and the statement, hence the return value of the expression
is not reflected in the state (it may only prevent a branch from being
taken). Communicating information from an expression back to a state
is a non-trivial problem which 
depends on the language considered, but we can help the
abstract interpretation by splitting the state in three parts:
\(\mset{\mpar{\jsid{x} \mapsto 0, t, \jsid{x} \mapsto 0}, \mpar{\jsid{x} \mapsto
    \itvl{1}{\infty}, t, \jsid{x} \mapsto 0}, \mpar{\jsid{x} \mapsto
    \itvl{0}{\infty}, t, \jsid{x} \mapsto 0}}\).
Let \(\abstripleset\) be the set of triples (listed below) obtained from adding triples 
for every sub-expression of \(t\).
We can show that \(\mset{\mpar{\jsid{x}
    \mapsto 0, t, \jsid{x} \mapsto 
    0}, \mpar{\jsid{x} \mapsto \itvl{1}{\infty}, t, \jsid{x} \mapsto 0}} \subset
\ainterpf{\abstripleset}\) (the second triple uses \(\mpar{\jsid{x} \mapsto
  \itvl{0}{\infty}, t, \jsid{x} \mapsto 0}\) to evaluate the recursive while
term). However there is still one of the three triples that cannot be
derived, \emph{viz.}, \(\mpar{\jsid{x} \mapsto
  \itvl{0}{\infty}, t, \jsid{x} \mapsto 0} \in \ainterpf{\abstripleset}\).

To derive this third triple, we introduce a proof technique called \emph{state splitting} to
obtain a more precise abstract semantics. The core idea of the technique is that
if the state \(\astate\) of a triple \(\mpar{\astate,\aterm,\aval}\) is covered
by the states of some triples
\(\mpar{\astate[1],\aterm,\aval}..\mpar{\astate[n],\aterm,\aval}\), in the sense
that \(\concr\funu{\astate} \subseteq \concr\funu{\astate[1]} \cup .. \cup
\concr\funu{\astate[n]}\), then we may use \(\mpar{\astate,\aterm,\aval}\) in
the input triple set \(\abstripleset\) of \(\ainterpf{\abstripleset}\)
\emph{without} having to show that \(\mpar{\astate,\aterm,\aval}\) is in the resulting
triple set \(\ainterpf{\abstripleset}\) and still remain correct.

Formally, we first define a function \(\stSplit{}\) from triple sets to triple
sets that adds such triples.

\begin{definition}
  Let \(\abstripleset\) an abstract triple set. We define the state splitting
  function \(\stSplit{\abstripleset}\) as:
    \[\stSplit{\abstripleset} =
    \mset[\mpar{\astate,\aterm,\aval}]{
    \begin{gathered}
      \mset{\mpar{\astate[1],\aterm,\aval}..\mpar{\astate[n],\aterm,\aval}}
      \subseteq \abstripleset \text{ with } n \geq 1\\
      \forall i \in [1..n]. \Sort[]{\astate} = \Sort[]{\astate[i]}\\
      \concr\funu{\astate} \subseteq \concr\funu{\astate[1]} \cup .. \cup
      \concr\funu{\astate[n]}
    \end{gathered}
    }\]
\end{definition}

\begin{lemma}
  For any \(\abstripleset\), \(\abstripleset\subseteq\stSplit{\abstripleset}\),
  and \(\stSplit{}\) is monotonic.
\end{lemma}

\begin{lemma}\label{lem:split-wf}
  Let \(\abstripleset\) a well-formed triple set, then
  \(\stSplit{\abstripleset}\) is well formed.
\end{lemma}
\begin{proof}
  Let \(\mpar{\astate,\aterm,\aval} \in \stSplit{\abstripleset}\), then there is
  some \(\astate[1],\aterm,\aval \in \abstripleset\) such that
  \(\Sort[]{\astate} = \Sort[]{\astate[1]} = \typin{\aterm}\) and
  \(\Sort[]{\aval} = \typout{\aterm}\).
\end{proof}

We next show that the functional \(\sinterpf{\cdot}\) has the same consistency
property as \(\ainterpf{\cdot}\).

\begin{lemma}\label{lem:split-concr-wf-consistent}
    Let \(\tripleset\) and \(\abstripleset\) be well formed and consistent triple
  sets, then
  \(\cinterpf{\tripleset}\) and \(\sinterpf{\abstripleset}\) are well formed and
  consistent triple sets.

\end{lemma}

We finally state that the proof technique is correct.
\begin{lemma}\label{lem:split-correct}
    Let \(\abstripleset\) a well-formed abstract triple set. If \(\abstripleset
  \subseteq \ainterpf{\stSplit{\abstripleset}}\), then
  \(\stSplit{\abstripleset}\) is correct.

\end{lemma}

We turn back to our example. Consider the following triple set.
\begin{equation*}
\abstripleset = 
\mset{
  \begin{gathered}
    \mpar{\jsid{x} \mapsto 0, 0, 0}, \mpar{\jsid{x} \mapsto \itvl{1}{\infty}, 0, 0}, 
    \mpar{\jsid{x} \mapsto 0, -1, -1}, \mpar{\jsid{x} \mapsto \itvl{1}{\infty}, -1, -1},\\
    \mpar{\jsid{x} \mapsto 0, \jsid{x}, 0},
    \mpar{\jsid{x} \mapsto \itvl{1}{\infty}, \jsid{x}, \itvl{1}{\infty}},\\
    \mpar{\jsid{x} \mapsto 0, \jsid{x} = 0, \abs{\btrue}},
    \mpar{\jsid{x} \mapsto \itvl{1}{\infty}, \jsid{x} = 0, \abs{\bfalse}},\\
    \mpar{\jsid{x} \mapsto 0, \neg\mpar{\jsid{x} = 0}, \abs{\bfalse}},
    \mpar{\jsid{x} \mapsto \itvl{1}{\infty}, \neg\mpar{\jsid{x} = 0}, \abs{\btrue}},\\
    \mpar{\jsid{x} \mapsto \itvl{1}{\infty}, \jsid{x} -1, \itvl{0}{\infty}},
    \mpar{\jsid{x} \mapsto \itvl{1}{\infty}, \asnescape{\jsid{x}}{\jsid{x} -1},
      \jsid{x} \mapsto \itvl{0}{\infty}},\\
    \mpar{\jsid{x} \mapsto 0, t, \jsid{x} \mapsto 0},
    \mpar{\jsid{x} \mapsto \itvl{1}{\infty}, t, \jsid{x} \mapsto 0}
  \end{gathered}
}
\end{equation*}
We can show that \(\abstripleset \subseteq \sinterpf{\abstripleset}\), hence
every triple of \(\stSplit{\abstripleset}\) is correct, in particular
\(\mpar{\jsid{x} \mapsto \itvl{0}{\infty}, t, \jsid{x} \mapsto 0}\).

Note that this proof technique does not depend on the programming
language considered. The difficulty is transferred to the choice of how to split
the state, but as long as the splitting is correct (the added triple is covered
by the existing ones), the resulting technique is sound.

\section{Constraint Generation}\label{SEC:CONSTRAINTS}

As a final interpretation, we show how the abstract
interpretation can be used to construct an actual program analyser. We
define the analyser as an interpretation that generates data flow
constraints to analyse a given
program~\cite{Nielson:Principles}.\footnote{We impose the technical
  restriction that any hook used in a skeleton can be matched to a program
  point of the program (closed term) \(t_0\) under consideration. Thus
  constraint-based analysis of code-generating code is not considered here.}
Constraint-based program analysis is a well-known technique for defining
analyses. We show how this technique can be lifted and defined
entirely as an interpretation, by generating constraints
over all the flow variables used in a semantic definition. 

We first need to formalise (and extend) the standard notion
of \emph{program point}. We take a program point \(\progpoint\)
to be a list of integers denoting a position in a term. Program points
form a monoid with concatenation operator \(\ppc{}{}\) and neutral
element \(\ppe\). We define a subterm operator
\(\subterm{t}{\progpoint}\) as follows.
\begin{align*}
  \subterm{t}{\ppe}
  & \eqdef t  \quad\quad\quad\quad
   \subterm{c\myvecip{t}{n}}{\ppc{k}{\progpoint}}
   \eqdef
    \begin{cases}
      \subterm{t_k}{\progpoint} & \text{if } k \in [1..n]\\
      \text{undefined} & \text{otherwise}
    \end{cases}
\end{align*}

We assume a function \(\Gent{}\) that for a given term \(t_0\) states
the set of program points for 
which constraints will be generated. It typically consists of the set of
executable subterms of \(t_0\). We require the program points of
\(\Gent{t_0}\) to be executable: if \(\progpoint \in
\Gent{t_0}\), then \(\subterm{t_0}{\progpoint} = c\myvecp{\term}{n}\).
Requirement~\ref{req:single-skeleton} enforces the existence of a skeleton
for this term.

We next define a partial operator \(\ppgen{}{}{}\) that associates program points
to the terms occurring 
in the hooks of a skeleton. Formally, if \(\ppgen{\progpoint}{\rulen{N}}{t} =
\progpoint'\), then (1) skeleton \(\rulen{N}\) is applicable: \(\progpoint \in
\Gent{t_0}\), \(\subterm{t_0}{\progpoint} = c\myvecp{\term}{n}\), and
\(\rulen{N}\) is of the form \(\Rule{N}{c\myvecp{\tvar}{n}}{\skel}\), (2) a
hook \(\deriv{\_}{\term}{\_}\) occurs in \(\skel\), and (3) the
resulting program point is part of the set of explored program points:
\(\progpoint' \in \Gent{t_0}\) and \(\subterm{\term_0}{\progpoint'} =
\mpar{\extsigvec{}{\tvar}{\term}}\funu{\term}\).

Constraints are either of the form \(\syntaxcontraint{x = x'}\),
\(\syntaxcontraint{x \sqsubseteq x'}\), or \(\syntaxcontraint{x : s}\),
where \(x\) and \(x'\) are variables and \(s\) a sort.
We generate variable names in constraints
of the form
\(\cppv{\progpoint}{x}\).
The constraint generation function $\Gen{}$
that takes a program \(t_0\) and returns the set of constraints
generated by \(t_0\) is defined as 
\begin{align*}
  \Gen{t_0} &\eqdef \bigcup \eqset \cup
  \mset{
              \begin{gathered}
                  \syntaxcontraint{\cppv{\progpoint}{\ivar} : \typin{\Sort[]{\subterm{t_0}{\progpoint}}}},\\
                  \syntaxcontraint{\cppv{\progpoint}{\ovar} : \typout{\Sort[]{\subterm{t_0}{\progpoint}}}},\\
                \forall i \in [1..n].
                  \syntaxcontraint{\cppv{\progpoint}{\cppv{}{\tvar[i]}} = \term[i]}
              \end{gathered}
  \quad\left|\quad
  \begin{gathered}
    \progpoint \in \Gent{t_0} \land
    \subterm{t_0}{\progpoint} = c\myvecp{\term}{n}\\
    \Rule{N}{c\myvecp{\tvar}{n}}{\skel} \in \Rules\\
    \eqgen[\rulen{N},\progpoint,\emptyset]{\skel}{\eqset}\\
    \eqdfvar{\emptyset} = \mset{\myvec{\tvar}{n}, \ivar} \land
    \ovar \in \eqdfvar{\eqset}
  \end{gathered}
  \right.
  }
\end{align*}

For each skeleton \(\rulen{N}\) we define a function \(\eqdfvar{}\) that maps sets
of constraints to sets of skeletal variables. This is not necessary for the
constraint generation but is used to prove consistency between constraints and
the abstract semantics.

The constraint generation interpretation of skeletons \(\eqgen[]{S}{}\) is given in
Figure~\ref{fig:constraint-generation}. The rule for hooks generates
constraints for connecting the input state \(\cppv{\progpoint'}{\ivar}
\) with the flow variable holding the input state in the hook
\(
\cppv{\progpoint}{\cppv{}{\fvar[1]}}
\),
and the resulting output state of the hook with the output of the
hook.
Each filter comes with a constraint generation function
\(\eqfilter{F}{}{}\) specific to the analysis of that filter.
We require that the constraints generated for that filter
agree with the abstract semantics: if
\(\solution\) is a solution to the constraints
\(\eqfilter{F}{\myvec{\cppv{\progpoint}{\cppv{}{\svar{x}}}}{n},
  \myvec{\cppv{\progpoint}{\cppv{}{\svar{y}}}}{m}}{}\), then
following holds:
\(\abstrfilter{F}{\myvec[\solution]{\cppv{\progpoint}{\cppv{}{\svar{x}}}}{n}}{} \sqsubseteq
\myvecp[\solution]{\cppv{\progpoint}{\cppv{}{\svar{y}}}}{m}\).
For analysing a set of branches, we generate constraints for each
branch and return the union of these constraint sets.

\begin{figure}
  \centering
  \begin{align*}
    &\implies \eqgen{\emptylist}{\eqset}\\
    \mpar{
    \begin{rgathered}
      \ppgen{\progpoint}{\rulen{N}}{t} = \progpoint'\\
      \fvar[1] \in \eqdfvar{\eqset}\\
      \eqset' = \eqset \cup \mset{
        \begin{gathered}
            \syntaxcontraint{\cppv{\progpoint}{\cppv{}{\fvar[1]}} \sqsubseteq \cppv{\progpoint'}{\ivar}},\\
            \syntaxcontraint{\cppv{\progpoint'}{\ovar} \sqsubseteq \cppv{\progpoint}{\cppv{}{\fvar[2]}}}
        \end{gathered}}\\
      \eqdfvar{\eqset'} = \eqdfvar{\eqset} \cup \mset{\fvar[2]}
    \end{rgathered}
}
    &\implies \eqgen{\deriv{\fvar[1]}{\term}{\fvar[2]}}{\mpar{\rulen{N},\progpoint,\eqset'}}\\
    \mpar{
    \begin{rgathered}
    \myvecs{\svar{x}}{n} \subseteq \eqdfvar{\eqset}\\
      \eqfilter{F}{
        \begin{gathered}
          \myvec{\cppv{\progpoint}{\cppv{}{\svar{x}}}}{n},\\
          \myvec{\cppv{\progpoint}{\cppv{}{\svar{y}}}}{m}
        \end{gathered}
}{\eqset_f}\\
\eqset' = \eqset \cup \eqset_f\\
\eqdfvar{\eqset'} = \eqdfvar{\eqset} \cup \myvecs{\svar{y}}{m}
    \end{rgathered}
}
    &\implies \eqgen{\filter{F}{\myvec{\svar{x}}{n}}{\myvecp{\svar{y}}{m}}}{\mpar{\rulen{N},\progpoint,\eqset'}}
      \\
    \mpar{
    \begin{rgathered}
      i \geq 1\\
      \forall i \in [1..n]. \outfun\funu{i} = \eqset_i\\
      \forall i \in [1..n]. \bvar \subseteq \eqdfvar{\eqset_i}\\
      \eqset' = \eqset \cup \bigcup_{i \in [1..n]} \eqset_i\\
      \eqdfvar{\eqset'} = \eqdfvar{\eqset} \cup \bvar
    \end{rgathered}}
    &\implies \eqgenb[\outfun,\mpar{\rulen{N},\progpoint,\eqset}]{\bigoplus_n}{\mpar{\rulen{N},\progpoint,\eqset'}}
  \end{align*}
  \caption{Constraint Generation}\label{fig:constraint-generation}
\end{figure}

\paragraph{Correctness}
A \emph{solution} \(\solution\) of a set of constraints \(\eqset\) is a mapping
from the variables in \(\eqset\) to abstract values and terms such that every
constraint in \(\eqset\) holds.

\begin{lemma}\label{lem:constraints-correct}
    Let \(\term[0]\) be a term and \(\solution\) be a solution of
  \(\Gen{\term[0]}\). Let
\(\abstripleset\) be defined as follows:
\begin{equation*}
  \abstripleset =  \mset{\mpar{\astate,\term,\aval} \st
  \begin{gathered}
    \progpoint \in \Gent{t_0}\\
    \term = \subterm{t_0}{\progpoint}\\
    \solution\funu{\cppv{\progpoint}{\ivar}} = \astate\\
    \solution\funu{\cppv{\progpoint}{\ovar}} = \aval
  \end{gathered}}
\end{equation*}
Then \(\abstripleset\) is well typed and \(\abstripleset \subseteq \ainterpf{\abstripleset}\).

\end{lemma}

\paragraph{Discussion} The constraints we generate are \emph{path-insensitive}:
they do not capture the fact that when a filter does not hold, the rest of the
skeleton does not matter. Constraints can be path-sensitive by letting the state
of the interpretation be a pair consisting of a set \(\mathit{Stop}\) of
constraint sets representing pathways in the skeleton that are stopped, similar
to the \(\halt\) flag in the abstract interpretation, and another set
\(\mathit{Run}\) of constraint sets representing all the running paths. When a filter is
encountered, the \(\mathit{Run}\) sets are added to \(\mathit{Stop}\) with the
additional constraint that the filter returns \(\bot\). The usual constraints
for the filter are added to each set in \(\mathit{Run}\). In a nutshell, we duplicate
constraints for each filter: once when it does not hold, and once when it may
hold. At the end of the interpretation, the global constraint to be satisfied is
the disjunction of all constraint sets in \(\mathit{Run}\) and \(\mathit{Stop}\), each constraint
set interpreted as a conjunction of its atomic constraints.

\paragraph{Example}
Consider \(t_0 = \while{\neg\mpar{\jsid{x} = 0}}{\asnescape{\jsid{x}}{\jsid{x} - 1}}\).
Its executable subterms are \[\Gent{t_0} =
    \mset{\ppe, 1, \ppc{1}{1}, \ppc{1}{\ppc{1}{1}}, \ppc{1}{\ppc{1}{2}}, 2, \ppc{2}{2}, \ppc{2}{\ppc{2}{1}}, \ppc{2}{\ppc{2}{2}}}.\]
Note that the subterm \(\jsid{x}\) appears both as program points
\(\ppc{1}{\ppc{1}{1}}\) and \(\ppc{2}{\ppc{2}{1}}\) of \(t_0\).
For the different filters we generate symbolic constraints that will
reuse abstract filters:
\(\eqfilter{\mathtt{isBool}}{x, y}{\mset{\syntaxcontraint{y = \mathtt{isBool}\funu{x}}}}\).
A mapping \(\solution\) is then a solution of such a symbolic constraint
if \(\solution\funu{y} = \abstrfilter{\mathtt{isBool}}{\solution\funu{x}}{}\).

The definition of \(\Gen{t_0}\) generates a large number of constraints.
We focus on a selection of them: those generated by the initial program point \(\ppe\).
The associated skeleton is
\[
\Rule{While}{\while{\tvar[1]}{\tvar[2]}}
     {\left[
       \deriv{\ivar{}}{\tvar[1]}{\fvar[1]};
       \filter{\texttt{isBool}}{\fvar[1]}{\fvar[1']};
       \ldots
       \right].}
\]
The constraint generation then produces the constraints
\begin{mathpar}
\syntaxcontraint{\cppv{\ppe}{\ivar} : \stateT}, \and
\syntaxcontraint{\cppv{\ppe}{\tvar[1]} =
    \neg\mpar{\jsid{x} = 0}}, \and
\syntaxcontraint{\cppv{\ppe}{\ovar} : \stateT}, \and
\syntaxcontraint{\cppv{\ppe}{\tvar[2]} =
  \asnescape{\jsid{x}}{\jsid{x} - 1}}.
\end{mathpar}
as well as the constraints given by \(\eqgen[\rulen{N},\progpoint,\emptyset]{
       \deriv{\ivar{}}{\tvar[1]}{\fvar[1]};
       \filter{\texttt{isBool}}{\fvar[1]}{\fvar[1']};
       \ldots
}{}\).
The hook case links the variable \(\cppv{\ppe}{\ivar}\)
to the input of \(\tvar[1]\), which here represents \(\neg\mpar{\jsid{x} = 0}\):
\(\ppgen{\ppe}{\rulen{While}}{\tvar[1]} = 1\)
and we thus generate the two constraints 
\begin{align*}
  \syntaxcontraint{\cppv{\ppe}{\ivar}
    \sqsubseteq \cppv{1}{\ivar}}, 
\syntaxcontraint{\cppv{1}{\ovar} \sqsubseteq
  \cppv{\ppe}{\fvar[1]}}.
\end{align*}
The constraints on \(\cppv{1}{\ivar}\) and \(\cppv{1}{\ovar}\) are generated
when considering the program point~\(1\),
corresponding to the evaluation of \(\neg\mpar{\jsid{x} = 0}\)
(corresponding to the skeleton \(\rulen{Neg}\)).
As stated, the set of all generated constraints is large; it is provided in the
supplementary material on the companion website.

\section{Extending \While{} with exceptions, input/output, and a heap}
\label{sec:while2}

To further illustrate the use of skeletal semantics, we extend our \While{} language with exceptions, input/output, and a heap. We
first need to define new flow sorts: \(\inT\) for input streams, \(\outT\) for
output streams, \(\heapT\) for heaps, \(\locT\) for locations in the heap, \(\IOstateT\) for the combination of the
streams with a store and a heap, \(\valIOstateT\) for the further combination
with a value, and \(\excIOstateT\) for a \(\IOstateT\) extended to signal
whether an exception was raised.
We still have two program sorts (\(\exprS\) and \(\statS\)), but their input
flow sorts are \(\IOstateT\), and their output flow sorts are now \(\valIOstateT\) for
expressions and \(\excIOstateT\) for statements.
Figure \ref{fig:example:constructors:while2} lists the additional
constructors of our language.
The additional filters are defined in Figure
\ref{fig:example:typing:filters:while2}, the rules for expressions in Figure
\ref{fig:example:skeletal:semantics:while2:expr}, and the rules for statements
in Figure \ref{fig:example:skeletal:semantics:while2:stat}. To help reading the
rules, flow variables have names related to their sorts: \(\sigma\) for
\(\IOstateT\), \(w\) for \(\valIOstateT\), \(v\) for \(\valT\), \(n\) for
\(\intT\), \(i\) for \(\inT\), and so on.

\begin{figure}
  \centering
  \begin{equation*}
    \begin{array}[t]{|c|c|}
      \hline
      \text{\(c\)} & \text{Signature} \\
      \hline
      \ein & \exprS\\
      \eoutn & \exprS \rightarrow \statS \\
      \hline
    \end{array}
    \quad
    \begin{array}[t]{|c|c|}
      \hline
      \text{\(c\)} & \text{Signature} \\
      \hline
      \sthrow & \statS\\
      \stry{}{} & \mpar{\statS,\statS} \rightarrow \statS\\
      \hline
    \end{array}
    \quad
    \begin{array}[t]{|c|c|}
      \hline
      \text{\(c\)} & \text{Signature} \\
      \hline
      \erefn & \exprS \rightarrow \exprS\\
      ! & \exprS \rightarrow \exprS\\
      \leftarrow & \mpar{\exprS,\exprS} \rightarrow \statS\\
      \hline
    \end{array}
  \end{equation*}
  \caption{Additional Constructors for \While{}}
  \label{fig:example:constructors:while2}
\end{figure}

\begin{figure}
  \centering
\begin{equation*}
\begin{array}[t]{|c|c|}
  \hline
  f & \typsig{f}\\
  \hline
  \texttt{in} & \inT \rightarrow \mpar{\valT,\inT} \\
  \texttt{alloc} & \mpar{\heapT, \valT} \rightarrow \mpar{\heapT, \locT}\\
  \texttt{locVal} & \locT \rightarrow \valT \\
  \texttt{isLoc} & \valT \rightarrow \locT \\
  \texttt{get} & \mpar{\locT,\heapT} \rightarrow \valT\\
  \texttt{set} & \mpar{\locT,\heapT,\valT} \rightarrow \heapT\\
  \texttt{out} & (\outT,\valT) \rightarrow \outT\\
  \hline
\end{array}
\quad
\begin{array}[t]{|c|c|}
  \hline
  f & \typsig{f}\\
  \hline
  \texttt{mkSt} & \mpar{\inT,\outT,\stateT,\heapT} \rightarrow \IOstateT \\
  \texttt{splitSt} & \IOstateT \rightarrow \mpar{\inT,\outT,\stateT,\heapT}\\
  \texttt{mkValSt} & \mpar{\valT,\IOstateT} \rightarrow \valIOstateT \\
  \texttt{getValSt} & \valIOstateT \rightarrow \mpar{\valT,\IOstateT}\\
  \texttt{mkOK} & \IOstateT \rightarrow \excIOstateT\\
  \texttt{mkExc} & \IOstateT \rightarrow \excIOstateT\\
  \texttt{isOK} & \excIOstateT \rightarrow \IOstateT\\
  \texttt{isExc} & \excIOstateT \rightarrow \IOstateT\\
  \hline
\end{array}
\end{equation*}
  \caption{Additional filters}\label{fig:example:typing:filters:while2}
\end{figure}

\begin{figure}
  \begin{align*}
    \Rulea{Lit}{\const{\tvar}}
    {\left[\filter{\texttt{litInt}}{\tvar}{\fvar[n]{}};
    \filter{\texttt{intVal}}{\fvar[n]}{\fvar[v]};
    \filter{\texttt{mkValSt}}{\fvar[v]{},\ivar{}}{\ovar{}}
    \right]}
    \\
    \Rulea{Var}{\var{\tvar}}
    {\left[
    \begin{multlined}[][\arraycolsep]
      \filter{\texttt{splitSt}}{\ivar{}}{\mpar{\fvar[i]{},\fvar[o]{},\fvar[s]{},\fvar[h]{}}};
      \filter{\texttt{read}}{\tvar{},\fvar[s]{}}{\fvar[v]{}};\\
      \filter{\texttt{mkSt}}{\fvar[i]{},\fvar[o]{},\fvar[s]{},\fvar[h]{}}{\fvar[\sigma]{}};
      \filter{\texttt{mkValSt}}{\fvar[v]{},\fvar[\sigma]{}}{\ovar{}}
    \end{multlined}
    \right]}
    \\
    \Rulea{In}{\ein}
    {\left[
    \begin{multlined}[][\arraycolsep]
      \filter{\texttt{splitSt}}{\ivar{}}{\mpar{\fvar[i]{},\fvar[o]{},\fvar[s]{},\fvar[h]{}}};
      \filter{\texttt{in}}{\fvar[i]{}}{\mpar{\fvar[v]{},\fvar[i']{}}};\\
      \filter{\texttt{mkSt}}{\fvar[i']{},\fvar[o]{},\fvar[s]{},\fvar[h]{}}{\fvar[\sigma]{}};
      \filter{\texttt{mkValSt}}{\fvar[v]{},\fvar[\sigma]{}}{\ovar{}}
    \end{multlined}
    \right]}
    \\
    \Rulea{Alloc}{\eref{\tvar}}
    {\left[
    \begin{multlined}[][\arraycolsep]
      \deriv{\ivar{}}{\tvar{}}{\fvar[w]{}};
      \filter{\texttt{getValSt}}{\fvar[w]{}}{\mpar{\fvar[v]{},\fvar[\sigma]{}}};\\
      \filter{\texttt{splitSt}}{\fvar[\sigma]{}}{\mpar{\fvar[i]{},\fvar[o]{},\fvar[s]{},\fvar[h]{}}};
      \filter{\texttt{alloc}}{\fvar[h]{},\fvar[v]{}}{\mpar{\fvar[h']{},\fvar[l]{}}};\\
      \filter{\texttt{locVal}}{\fvar[l]}{\fvar[v']};
      \filter{\texttt{mkSt}}{\fvar[i]{},\fvar[o]{},\fvar[s]{},\fvar[h']{}}{\fvar[\sigma']{}};\\
      \filter{\texttt{mkValSt}}{\fvar[v']{},\fvar[\sigma']{}}{\ovar{}}
    \end{multlined}
    \right]}
    \\
    \Rulea{Acc}{\access{\tvar}}
    {\left[
    \begin{multlined}[][\arraycolsep]
      \deriv{\ivar{}}{\tvar{}}{\fvar[w]{}};
      \filter{\texttt{getValSt}}{\fvar[w]{}}{\mpar{\fvar[v]{},\fvar[\sigma]{}}};
      \filter{\texttt{isLoc}}{\fvar[v]{}}{\fvar[l]};\\
      \filter{\texttt{splitSt}}{\fvar[\sigma]{}}{\mpar{\fvar[i]{},\fvar[o]{},\fvar[s]{},\fvar[h]{}}};
      \filter{\texttt{get}}{\fvar[l]{},\fvar[h]{}}{\fvar[v']{}};\\
      \filter{\texttt{mkSt}}{\fvar[i]{},\fvar[o]{},\fvar[s]{},\fvar[h]{}}{\fvar[\sigma']{}};
      \filter{\texttt{mkValSt}}{\fvar[v']{},\fvar[\sigma']{}}{\ovar{}}
    \end{multlined}
    \right]}
    \\
      \Rulea{Add}{\tvar[1]{} + \tvar[2]{}}
      {\left[
      \begin{multlined}[][\arraycolsep]
        \deriv{\ivar{}}{\tvar[1]{}}{\fvar[w_{1}]{}};
        \filter{\texttt{getValSt}}{\fvar[w_{1}]{}}{\mpar{\fvar[v_{1}]{},\fvar[\sigma_{1}]{}}};
        \filter{\texttt{isInt}}{\fvar[v_{1}]{}}{\fvar[n_{1}]};\\
        \deriv{\fvar[\sigma_{1}]{}}{\tvar[2]{}}{\fvar[w_{2}]{}};
        \filter{\texttt{getValSt}}{\fvar[w_{2}]{}}{\mpar{\fvar[v_{2}]{},\fvar[\sigma_{2}]{}}};
        \filter{\texttt{isInt}}{\fvar[v_{2}]{}}{\fvar[n_{2}]{}};\\
        \filter{\texttt{add}}{\fvar[n_{1}]{}, \fvar[n_{2}]{}}{\fvar[n]{}};
        \filter{\texttt{intVal}}{\fvar[n]}{\fvar[v]};
        \filter{\texttt{mkValSt}}{\fvar[v]{},\fvar[\sigma_{2}]{}}{\ovar{}}
      \end{multlined}
      \right]}
    \\
      \Rulea{Eq}{\tvar[1]{} = \tvar[2]{}}
      {\left[
      \begin{multlined}[][\arraycolsep]
        \deriv{\ivar{}}{\tvar[1]{}}{\fvar[w_{1}]{}};
        \filter{\texttt{getValSt}}{\fvar[w_{1}]{}}{\mpar{\fvar[v_{1}]{},\fvar[\sigma_{1}]{}}};
        \filter{\texttt{isInt}}{\fvar[v_{1}]{}}{\fvar[n_{1}]};\\
        \deriv{\fvar[\sigma_{1}]{}}{\tvar[2]{}}{\fvar[w_{2}]{}};
        \filter{\texttt{getValSt}}{\fvar[w_{2}]{}}{\mpar{\fvar[v_{2}]{},\fvar[\sigma_{2}]{}}};
        \filter{\texttt{isInt}}{\fvar[v_{2}]{}}{\fvar[n_{2}]{}};\\
        \filter{\texttt{eq}}{\fvar[n_{1}]{}, \fvar[n_{2}]{}}{\fvar[b]{}};
        \filter{\texttt{boolVal}}{\fvar[b]}{\fvar[v]};
        \filter{\texttt{mkValSt}}{\fvar[v]{},\fvar[\sigma_{2}]{}}{\ovar{}}
      \end{multlined}
      \right]}
    \\
      \Rulea{Neg}{\neg{\tvar}}
      {\left[
    \begin{multlined}[][\arraycolsep]
      \deriv{\ivar{}}{\tvar}{\fvar[w]};
      \filter{\texttt{getValSt}}{\fvar[w]{}}{\mpar{\fvar[v]{},\fvar[\sigma]{}}};
      \filter{\texttt{isBool}}{\fvar[v]}{\fvar[b]};\\
      \filter{\texttt{neg}}{\fvar[b]}{\fvar[b']{}};
      \filter{\texttt{boolVal}}{\fvar[b']}{\fvar[v']};
      \filter{\texttt{mkValSt}}{\fvar[v']{},\fvar[\sigma]{}}{\ovar{}}
      \end{multlined}
    \right]}
  \end{align*}
  \caption{Skeletal semantics for extended \While{} (Expressions)}\label{fig:example:skeletal:semantics:while2:expr}
\end{figure}

\begin{figure}
  \begin{align*}
      \Rulea{Skip}{\sskip}
      {\left[ \filter{\texttt{mkOK}}{\ivar{}}{\ovar{}} \right]}
    \quad
      \Rule{Throw}{\sthrow}
      {\left[ \filter{\texttt{mkExc}}{\ivar{}}{\ovar{}} \right]}
    \\
      \Rulea{Asn}{\asnescape{\tvar[1]}{\tvar[2]}}
      {\left[
    \begin{multlined}[][\arraycolsep]
      \deriv{\ivar{}}{\tvar[2]}{\fvar[w]};
      \filter{\texttt{getValSt}}{\fvar[w]{}}{\mpar{\fvar[v]{},\fvar[\sigma]{}}};\\
      \filter{\texttt{splitSt}}{\fvar[\sigma]{}}{\mpar{\fvar[i]{},\fvar[o]{},\fvar[s]{},\fvar[h]{}}};
      \filter{\texttt{write}}{\tvar[1]{},\fvar[s]{},\fvar[v]}{\fvar[s']{}};\\
      \filter{\texttt{mkSt}}{\fvar[i]{},\fvar[o]{},\fvar[s']{},\fvar[h]}{\fvar[\sigma']};
      \filter{\texttt{mkOK}}{\fvar[\sigma']{}}{\ovar{}}
    \end{multlined}
    \right]}
    \\
      \Rulea{Set}{\heapwrite{\tvar[1]}{\tvar[2]}}
      {\left[
    \begin{multlined}[][\arraycolsep]
      \deriv{\ivar{}}{\tvar[1]}{\fvar[w_{1}]};
      \filter{\texttt{getValSt}}{\fvar[w_{1}]{}}{\mpar{\fvar[v_{1}]{},\fvar[\sigma]{}}};
      \filter{\texttt{isLoc}}{\fvar[v_{1}]}{\fvar[l]}\\
      \deriv{\fvar[\sigma]{}}{\tvar[2]}{\fvar[w_{2}]};
      \filter{\texttt{getValSt}}{\fvar[w_{2}]{}}{\mpar{\fvar[v_{2}]{},\fvar[\sigma']{}}};\\
      \filter{\texttt{splitSt}}{\fvar[\sigma']{}}{\mpar{\fvar[i]{},\fvar[o]{},\fvar[s]{},\fvar[h]{}}};
      \filter{\texttt{set}}{\fvar[l]{},\fvar[h]{},\fvar[v_{2}]}{\fvar[h']{}};\\
      \filter{\texttt{mkSt}}{\fvar[i]{},\fvar[o]{},\fvar[s]{},\fvar[h']}{\fvar[\sigma'']};
      \filter{\texttt{mkOK}}{\fvar[\sigma'']{}}{\ovar{}}
    \end{multlined}
    \right]}
    \\
      \Rulea{Out}{\eout{\tvar[1]}}
      {\left[
    \begin{multlined}[][\arraycolsep]
      \deriv{\ivar{}}{\tvar[1]}{\fvar[w]};
      \filter{\texttt{getValSt}}{\fvar[w]{}}{\mpar{\fvar[v]{},\fvar[\sigma]{}}};\\
      \filter{\texttt{splitSt}}{\fvar[\sigma]{}}{\mpar{\fvar[i]{},\fvar[o]{},\fvar[s]{},\fvar[h]{}}};
      \filter{\texttt{out}}{\fvar[o]{},\fvar[v]{}}{\fvar[o']{}};\\
      \filter{\texttt{mkSt}}{\fvar[i]{},\fvar[o']{},\fvar[s]{},\fvar[h]}{\fvar[\sigma']};
      \filter{\texttt{mkOK}}{\fvar[\sigma']{}}{\ovar{}}
    \end{multlined}
    \right]}
    \\
      \Rulea{Seq}{\seq{\tvar[1]}{\tvar[2]}}
      {\left[
    \deriv{\ivar{}}{\tvar[1]}{\fvar[e]{}};
    \branchesV[\mset{\ovar{}}]{
    \begin{aligned}
      &\filter{\texttt{isOK}}{\fvar[e]{}}{\fvar[\sigma]{}};
      \deriv{\fvar[\sigma]}{\tvar[2]}{\ovar{}} \\
      &\filter{\texttt{isExc}}{\fvar[e]{}}{\fvar[\sigma']{}}; \filter{\texttt{mkExc}}{\fvar[\sigma']}{\ovar}
    \end{aligned}
    }
    \right]}
    \\
      \Rulea{Try}{\stry{\tvar[1]}{\tvar[2]}}
      {\left[
    \deriv{\ivar{}}{\tvar[1]}{\fvar[e]{}};
    \branchesV[\mset{\ovar{}}]{
    \begin{aligned}
      &\filter{\texttt{isOK}}{\fvar[e]{}}{\fvar[\sigma]{}}; \filter{\texttt{mkOK}}{\fvar[\sigma]}{\ovar}\\
      &\filter{\texttt{isExc}}{\fvar[e]{}}{\fvar[\sigma']{}};
      \deriv{\fvar[\sigma']}{\tvar[2]}{\ovar{}}
    \end{aligned}
    }
    \right]}
    \\
      \Rulea{If}{\sif{\tvar[1]}{\tvar[2]}{\tvar[3]}}
      {\left[
    \begin{multlined}[][\arraycolsep]
      \deriv{\ivar{}}{\tvar[1]}{\fvar[w]};
      \filter{\texttt{getValSt}}{\fvar[w]{}}{\mpar{\fvar[v]{},\fvar[\sigma]{}}};\\
      \filter{\texttt{isBool}}{\fvar[v]}{\fvar[b]};
      \branchesV[\mset{\ovar{}}]{\begin{aligned}
          &\filter{\texttt{isTrue}}{\fvar[b]}{};
          \deriv{\fvar[\sigma]{}}{\tvar[2]}{\ovar{}}\\
          &\filter{\texttt{isFalse}}{\fvar[b]}{};
          \deriv{\fvar[\sigma]{}}{\tvar[3]}{\ovar{}}
        \end{aligned}}
    \end{multlined}
            \right]}
    \\
      \Rulea{While}{\while{\tvar[1]}{\tvar[2]}}
      {\left[
      \begin{multlined}[][\arraycolsep]
        \deriv{\ivar{}}{\tvar[1]}{\fvar[w]};
      \filter{\texttt{getValSt}}{\fvar[w]{}}{\mpar{\fvar[v]{},\fvar[\sigma]{}}};
        \filter{\texttt{isBool}}{\fvar[v]}{\fvar[b]};\\
        \branchesV[\mset{\ovar{}}]{\begin{aligned}
                &
                \begin{multlined}[][\arraycolsep]
                  \filter{\texttt{isTrue}}{\fvar[b]}{};
                  \deriv{\fvar[\sigma]{}}{\tvar[2]}{\fvar[e]};\\
                  \branchesV[\mset{\ovar{}}]{
                    \begin{aligned}
                      &\filter{\texttt{isOK}}{\fvar[e]}{\fvar[\sigma']};
                      \deriv{\fvar[\sigma']}{\while{\tvar[1]}{\tvar[2]}}{\ovar{}}\\
                      &\filter{\texttt{isExc}}{\fvar[e]}{\fvar[\sigma'']};
                      \filter{\texttt{mkExc}}{\fvar[\sigma'']}{\ovar{}}
                    \end{aligned}
                  }
                \end{multlined}
                \\
                &\filter{\texttt{isFalse}}{\fvar[b]}{};
                \filter{\texttt{mkOK}}{\fvar[\sigma]{}}{\ovar{}}
          \end{aligned}
        }
      \end{multlined}
                                   \right]}
  \end{align*}
  \caption{Skeletal semantics for \While{2} (Statements)}\label{fig:example:skeletal:semantics:while2:stat}
\end{figure}

\paragraph{Instantiation of concrete interpretation.} We instantiate the \(\inT\) and \(\outT\)
sorts with list of values, denoted by \(L\). We instantiate locations as
integers. A heap is a pair of an integer (the next free location) and a map
from integers to values. We instantiate the \(\IOstateT\) sort as a tuple of
\(\inT\), \(\outT\), \(\stateT\), and \(\heapT\), the \(\valIOstateT\) sort as
a pair of \(\valT\) and \(\IOstateT\), and the \(\excIOstateT\) sort as a pair
of a Boolean and \(\IOstateT\). The \(\valT\) sort is extended to include a case
for locations, as well as the \texttt{intVal}, \texttt{boolVal},
\texttt{isInt}, and \texttt{isBool} filters. The \texttt{locVal} filter
injects a location in the \(\valT\) type, and the \texttt{isLoc} filter applies
if the \(\valT\) argument is a location, which it then returns. The \texttt{in}
filter applies if the input list is not empty, it returns its head and its tail.
The \texttt{alloc} filters applied to \(\mpar{\mpar{n,m},v}\) returns the heap
\(\mpar{n+1,m+n \mapsto v}\). The \texttt{get} filter applies if the location is
in the heap, and it returns the corresponding value. The \texttt{set} filter
applies if the location is in the heap, and it return the heap updated with the
given value. The \texttt{out} filter always apply and adds the given value to
the output list. The \texttt{mkOK} filter (resp. the \texttt{mkExc} filter)
always applies and builds a pair of \textit{true} (resp. \textit{false}) and the
given state. The \texttt{isOK} filter (resp. the \texttt{isExc} filter) applies
if the Boolean is true (resp. is false), it then return the \(\IOstateT\)
component of the tuple. Other filters build or deconstruct tuples.

\paragraph{Instantiation of abstract interpretation.} To illustrate the
flexibility of our approach, we choose a coarse abstraction for the \(\inT\) and
\(\outT\) sorts (they are either \(\abot{}\) or a single abstract value), and a
precise abstraction of heaps: abstract heaps are modelled similar to concrete
heap as a pair \(\mpar{n,\abs{m}}\) of an integer and a mapping from integers to
abstract values. Locations are abstracted as sets of integers. Tuples are
abstracted as tuples of the abstraction of their components. The
tuple-manipulating abstract filters are straightforward, so we only detail the
other ones in Figure \ref{fig:abstract:while2}.

\begin{figure}
  \begin{align*}
    \gamma\funu{\abs{v}} &= \mset[L]{\forall v \in L. v \in \gamma\funu{\abs{v}}}
    & \gamma\funu{\mpar{\abs{a},\abs{b},\abs{c}}} &= \mset[\mpar{a,b,c}]{ a \in
                                                    \gamma\funu{a} \land b \in
                                                    \gamma\funu{b} \land c \in
                                                    \gamma\funu{c} } \\
    \gamma\funu{\abs{l}} &= \abs{l}
    & \gamma\funu{\mpar{n,\abs{m}}} &=
                                      \mset[\mpar{n,m}]{
                                      \begin{gathered}
                                        \dom{m} = \dom{\abs{m}} \\
                                        \forall i \in \dom{m}. m[i] \in
                                        \gamma\funu{\abs{m}[i]}
                                      \end{gathered}
    }
  \end{align*}

\begin{align*}
  \abs{\texttt{in}}\funu{\abs{v}}
&= \mpar{\abs{v},\abs{v}}
& \abs{\texttt{isOK}}\funu{\abs{b},\abs{s}}
&= \begin{cases} \abs{s}
  & \text{if } \abs{\mathit{true}} \sqsubseteq \abs{b}\\
  \abot{\IOstateT}
  &
  \text{otherwise}
\end{cases}\\
  \abs{\texttt{out}}\funu{\abs{v_{1}},\abs{v_{2}}}
 &= \abs{v_{1}} \sqcup \abs{v_{2}}
& \abs{\texttt{isExc}}\funu{\abs{b},\abs{s}}
  &= \begin{cases}
    \abs{s}
 & \text{if } \abs{\mathit{false}} \sqsubseteq \abs{b}\\
 \abot{\IOstateT}
 & \text{otherwise}
\end{cases}\\
  \abs{\texttt{locVal}}\funu{\abs{l}}
  &= \mpar{\abot{\intT},\abot{\boolT},\abs{l}}
  & \abs{\texttt{isval}}\funu{\abs{n},\abs{b},\abs{l}}
 &= \abs{l}\\
  \abs{\texttt{get}}\funu{\abs{l},\mpar{n,\abs{m}}}
 &= \bigsqcup_{l \in \abs{l}} \abs{m}[l]
& \abs{\texttt{alloc}}\funu{\mpar{n,\abs{m}},\abs{v}}
  &=
    \mpar{n+1,
    \abs{m}
    +
    n
    \mapsto \abs{v}}\\
  \abs{\texttt{set}}\funu{\abs{l},\mpar{n,\abs{m}},\abs{v}}
 &= \mathrlap{\mpar{n,\abs{m'}} \text{where} \begin{cases}
     \abs{m'}[l] = \abs{m}[l] \sqcup \abs{v} & \text{if } l \in \abs{l}\\
     \abs{m'}[l] =
     \abs{m}[l] &
     \text{otherwise}
   \end{cases}}
\end{align*}
\caption{Abstract Interpretation of Extended \While{}}\label{fig:abstract:while2}
\end{figure}

\begin{lemma}
  The abstract filters are consistent with the concrete filters.
\end{lemma}

\begin{lemma}
  The abstract semantics of Extended \While{} is correct.
\end{lemma}

\section{Related work}


Ott~\cite{Sewell:10:OTT} is a formalism for describing language
semantics and type systems. Ott proposes a meta-language with a
humanly readable syntax for writing semantic definitions as inference
rules, and has facilities for translating these definitions into
executable interpreters and specifications in proof assistants such as
Coq and HOL.
Lem~\cite{Mulligan:14:Lem} offers a core functional language extended
with logical features from proof assistants for writing
semantic models.
Ott can be used to describe static type systems but neither Ott nor
Lem has been used to derive program analyses.   

Action Semantics~\cite{Mosses:ActionSemantics} was developed by Mosses
and Watt as a modular format for writing semantics.  Turi and
Plotkin~\cite{turi1997towards} propose a generic way of defining small
step operational semantics, presented in a category theoretic
framework. More recently, Churchill emph{et
  al.}~\cite{churchill2015reusable} addresses the issue of the
reusability of operational semantics.  Their approach is based on
structures called \emph{fundamental constructs}, or \emph{funcons},
which only specify the changed parts of the state for a given
construct.  For instance, the funcon for \(\sif{}{}{}\) does not
mention environments, but only the Boolean part of the value which is
needed.  Funcons can then be combined to build a programming language.
There is a connection between these funcons and our rules as they are
both meant to capture the whole behaviour of a given language
construct. One difference is that funcons have a certain degree of
sort polymorphism, menaning that \emph{e.g.}, conditional statements
and conditional expressions can be treated by the same
``if''-funcon. Skeletal semantics would treat each sort separately but
would re-use the filters, so as to avoid increasing the proof effort
required.  To the extent of our knowledge, the work on funcons has
been focused on building extendable concrete semantics, and has never
been used to build an abstract semantics.


Views~\cite{Dinsdale-Young:13:Views} has a concrete operational
semantics for control flow,
but is parameterised on the state model and basic
commands. It proposes a program logic for this language, which is
parameterised on the actions of the basic commands.
They prove a general soundness result
stating that it suffices to check soundness for each basic command.
This corresponds in our framework
to the fact that only simple properties on filters need to be checked. Similarly,
Keidel \emph{et al.} \cite{Keidel:2018:CSP:3243631.3236767} very recently proposed to capture
the similarity between a concrete an abstract interpreter using a shared
interpreter parameterised by \emph{arrows} that could be instantiated to concrete
and abstract versions, thus reducing the proof effort needed.

Iris~\cite{jung2017iris} is a concurrent separation logic framework. It is
parameterised by a small-step reduction relation and
it proposes a logic to reason about resources.
This logic is parameterised by a representation of resources
in the form of an algebraic structure called a ``camera''. 
Cameras come with local properties about resources that users have to check.
These local constraints yield the soundness of the Iris logic.
To be used in practice,
Iris requires its users to provides lemmas about weakest preconditions
for each language construct.
These lemmas are easy to find and prove in simple examples 
(such as a vanilla \While{}),
but they require a deep understanding about how one reasons
about the considered language.
Such lemmas can be much complex to express (let alone prove)
in complex languages such as JavaScript~\cite{gardner2012towards}.
We believe that our framework guides the proof effort
by local reasoning: at each step, the abstract interpretation
naturally considers every applicable branch.

The \K{} framework~\cite{rosu-2010-jlap} proposes a formalism for writing
operational semantics and for constructing program verifiers directly
on top of the semantic definitions, as opposed to using an intermediate
representation and/or a verification generator linked to a specific
program logic. The semantic rules are given as rewriting rules over
terms of semantic state.  The \K{} framework has been used to write
semantic definitions of several real-world languages, including C,
Java, and JavaScript.  The program verifiers are based on matching
logic~\cite{rosu-2017-lmcs}, a formalism for reasoning about patterns
and the set of terms that they match. A language-independent set of
proof rules defines a Reachability Logic which can reason about the
set of reachable states of a program. This has been instantiated to
obtain program verifiers reasoning about data structures of
heap-manipulating programs in C, Java, and
JavaScript~\cite{stefanescu-2016-oopsla}.

The \K{} framework has goals similar to ours: derive verifiers from
operational semantics, correct by construction.
A key difference is that the semantics of the \K{} specification tool
is complex and not clearly documented~\cite{li2018isak}.
In this work, we have focused on crystallising a general yet simple
rule format.
Our format enables a general definition of when a semantics is
well-defined and provides a generic correctness theorem for the
derived program verifiers that can be machine-checked in the Coq proof
assistant.



Schmidt initiated the abstract interpretation of big-step operational
semantics~\citep{Schmidt:95:Natural} by showing how to abstract
derivation trees (using co-induction to harness infinite derivations)
and derived classical data flow and control flow analyses as
abstract interpretations. 
Other systematic derivations of static analyses have taken small-step
operational semantics as starting point.
Schmidt~\citep{Schmidt:97:Smallstep} discusses the general principles for such
an approach and compares small-step and big-step operational semantics
as foundations for abstract
interpretation.
Cousot~\citep{Cousot:98:Marktoberdorf} shows how to derive static
analyses for an imperative language defined by a compositional
transition semantics using the principles of abstract
interpretation. Midtgaard and Jensen~\citep{MidtgaardJ:08} use a
similar approach for calculating control-flow analyses for functional
languages from operational semantics in the form of abstract machines.
Van Horn and
Might~\citep{VanHorn:10:Abstracting,VanHorn:11:Abstracting} show how a
series of analyses for higher-order functional languages can be
derived from operational semantics formulated as abstract
machines. The atomic operations of the machines are given an abstract
interpretation and it is shown that the ``abstract abstract machines''
can simulate all the transitions of the concrete abstract machine.  The
abstract machines used by Van Horn and Might can be expressed in our
rule format: the atomic operations correspond to our filters and the
simulation result 
corresponds to our consistency
result for concrete and abstract interpretations. The two works differ
slightly in scope in that we are interested in a general semantic rule
format and its meta-theory whereas Van Horn and Might are concerned
with giving a systematic derivation of advanced analyses for
higher-order languages with state.

Inspired by Schmidt, Bodin \emph{et al.}~\citep{cpp2015} identify a
rule format that can be systematically instantiated to both
concrete and abstract semantics, with a generic  consistency result.
Our work generalises their approach.  Their rule format is based on a
non-standard style of operational semantics, called pretty-big-step
operational semantics~\cite{chargueraud2013pretty}, which cuts up
standard big-step rules into many fine-grained rules. 
In our work, one skeleton describes the behaviour of one language construct.
Our skeletal semantics captures many forms of traditional operational
semantics, such as the traditional big-step semantics studied in this paper.

%


\section{Conclusions and Future Work}

We have introduced a new meta-language for capturing the behaviour of programming languages,
called {\em skeletal semantics}.  A {skeleton} provides a simple
way of describing the {complete} behaviour of a language construct
\emph{in a single definition}. We have given a language-independent,
generic definition of {\em interpretation} of a skeletal semantics,
systematically deriving semantic judgements from the skeletons. We 
have explored four such interpretations: a well-formedness
interpretation; a concrete interpretation; an abstract interpretation;
and a constraint generator for flow-sensitive analysis. 
A key advantage of skeletal
semantics is that we are able to establish general,
language-independent  consistency results,
which can then be instantiated to specific 
programming language by proving simple language-dependent  filter lemmas.

In this paper, we have focused on proving   the fundamental properties of
skeletal semantics, using the simple   \While{} language and its
extensions as illustrative examples.  
We have demonstrated that
 we can capture many language constructs 
including higher-order and object-oriented features. In future, we
would like to explore how our formalism scales to 
real-world languages such as OCaml and JavaScript. 
For instance, the specification of JavaScript~\cite{ECMA2018} is
written in a style where the whole
behaviour of each language construct is described in  a single
definition. It should  be comparatively 
straightforward to provide a specification of JavaScript using   a skeletal semantics.

A distinguishing feature of skeletal semantics is that interpretations
can be used to characterise several styles of semantics, independently
of the language considered. In this paper, we have focused on big-step
semantics.  In future, we plan to capture other forms of semantics 
such as small-step operational semantics, semantics for describing  concurrent,
distributed and interactive computation, and abstract machines, using
an approach similar to~\cite{DBLP:journals/corr/Uustalu13}.

We have interesting proof techniques that are worth
further exploration.  For example, we  have demonstrated  how to add an abstract rule
for state splitting. The proof technique used to validate
this abstract rule, namely that abstract interpretation is a greatest
fixpoint, is not specific to state splitting. We thus  want  to explore
other abstract rules validated by this greatest fixpoint. In
particular, we conjecture that we can use this approach to obtain a
frame rule for skeletal semantics, paving the way for the integration
of separation logic as an abstract interpretation.  It is also
possible to generate better (more precise) constraints than those
given in Section~\ref{SEC:CONSTRAINTS}, as well as constraints for
other analyses  such as control flow analysis. Based on the skeletal
semantics for the \(\lambda\) calculus, we have reproduced the constraint
generation for 0-CFA \cite{Palsberg:Toplas:Closureanalysis}. We are
currently studying how more advanced control flow analyses for other
languages can be expressed in our framework.

Finally, we have mechanised in Coq the definitions of skeletal
semantics and interpretations, and have proved the  general
consistency results. We
have formalised the well-formedness, concrete and abstract
interpretations, verifying that  the abstract interpretation for
the \While{} language is correct. We have also mechanised a skeletal
semantics for the \(\lambda\) calculus. We are currently studying how
to leverage this Coq mechanisation to build a certificate checker for
abstract analysis.


\newpage

\bibliography{biblio}

%

\end{document}
